\documentclass[11pt,a4paper]{article}

\usepackage{microtype}
\usepackage{a4wide}
\usepackage{fullpage}
\usepackage[T1]{fontenc}
\usepackage[utf8]{inputenc}
\usepackage{authblk}

\usepackage[table]{xcolor}
\definecolor{lilla}{HTML}{750787}
\usepackage{amssymb,amsmath,amsthm}
\usepackage{nicefrac}
\usepackage{hyperref}

\usepackage[capitalise]{cleveref}
\usepackage{graphicx}
\usepackage{caption}
\usepackage{subcaption}
\usepackage{booktabs}
\usepackage{tabularx}
\usepackage{multicol}
\usepackage{gensymb}
\usepackage{enumerate}   
\usepackage{csquotes}
\usepackage[hyperpageref]{backref}
\usepackage{paralist}
\usepackage{etoolbox}
\usepackage[ruled,vlined,linesnumbered]{algorithm2e}
\usepackage{multirow} 
\renewcommand*{\backref}[1]{}
\renewcommand*{\backrefalt}[4]{\ifcase #1\or [p.~#2]\else [pp.~#2]\fi }

\hypersetup{
	colorlinks,citecolor=green!50!black,linkcolor=red!50!black,urlcolor=black
}

\usepackage[numbers]{natbib}
\usepackage{tikz}
\usetikzlibrary{shapes.geometric,calc}
\tikzstyle{xnode}=[circle,scale=1,draw,fill=white]

\usepackage{mathtools}
\usepackage{xargs}
\newcommandx{\set}[2][1=1]{\ensuremath{\{#1,\ldots,#2\}}}
\newcommandx{\tlog}[3][1=,3=]{\log_{#1}^{#3}(#2)}
\newcommandx{\ith}[2][1=th]{#2\nobreakdash-#1}

\newtheorem{theorem}{Theorem}
\newtheorem{lemma}{Lemma}
\newtheorem{observation}[lemma]{Observation}
\newtheorem{corollary}[theorem]{Corollary}

\crefname{observation}{Observation}{Observations}
\crefname{theorem}{Theorem}{Theorems}
\Crefname{theorem}{Thm.}{Thms.}
\crefname{corollary}{Corollary}{Corollaries}
\crefname{lemma}{Lemma}{Lemmata}
\Crefname{corollary}{Cor.}{Cors.}
\crefname{proposition}{Proposition}{Propositions}
\Crefname{proposition}{Prop.}{Props.}
\crefname{algorithm}{Algorithm}{Algorithms}
\crefname{algocf}{alg.}{algs.}
\Crefname{algocf}{Algorithm}{Algorithms}
\crefname{algocfline}{line}{lines}
\Crefname{algocfline}{l.}{l.}

\newcommand{\calF}{\mathcal{F}}
\newcommand{\calI}{\mathcal{I}}
\newcommand{\calT}{\mathcal{T}}

\newcommand{\cC}{\mathcal{C}}

\crefname{problem}{Problem}{Problems}
\Crefname{problem}{Prob.}{Probs.}

\newcommand{\boxproblem}[4]{
	\begin{center}   
		\fbox{~\begin{minipage}{.97\textwidth}
				\vspace{2pt} 
				\noindent
				\normalsize\textsc{#1}
				\vspace{1pt}
				
				\setlength{\tabcolsep}{3pt}
				\renewcommand{\arraystretch}{1.0}
				\begin{tabularx}{\textwidth}{@{}lX@{}}
					\normalsize\textbf{Input:}       & \normalsize#2 \\
					\normalsize\textbf{Output:}    & \normalsize#3 \\
					\normalsize\textbf{Objective:}    & \normalsize#4 \\
				\end{tabularx}
		\end{minipage}}
	\end{center}
}

\usepackage{dsfont}

\newcommand{\Qpos}{\mathbb{Q}_{\ge 0}}

\newcommand{\bigO}{\mathcal{O}}
\newcommand{\Oh}{\bigO}

\newcommand{\cocl}[1]{\ensuremath{\operatorname{#1}}}
\newcommand{\classP}{\cocl{P}}
\newcommand{\NP}{\cocl{NP}}

\newcommand{\calC}{\mathcal{C}}

\newcommand{\calP}{\mathcal{P}}
\newcommand{\calQ}{\mathcal{Q}}

\newcommand{\eps}{\varepsilon}

\newcommand{\prob}[1]{{\normalfont\textsc{#1}}\xspace}
\newcommand{\vertexcover}{\prob{Vertex Cover}}
\newcommand{\convertexcover}{\prob{Connected Vertex Cover}}
\newcommand{\weightedvertexcover}{\prob{Weighted Vertex Cover}}
\newcommand{\chromaticnumber}{\prob{Chromatic Number}}
\newcommand{\trianglepacking}{\prob{Triangle Packing}}

\DeclareMathOperator{\OPT}{OPT}
\DeclareMathOperator{\OPTVC}{OPT_{VC}}
\DeclareMathOperator{\OPTWVC}{OPT_{WVC}}
\DeclareMathOperator{\OPTCVD}{OPT_{CVD}}

\DeclareMathOperator{\OPTCGVD}{OPT_{CGVD}}
\DeclareMathOperator{\OPTCCVD}{OPT_{CCVD}}
\DeclareMathOperator{\OPTIVD}{OPT_{IVD}}
\DeclareMathOperator{\OPTCOL}{OPT_{COL}}
\DeclareMathOperator{\OPTFVS}{OPT_{FVS}}
\DeclareMathOperator{\OPTTP}{OPT_{TP}}

\DeclarePairedDelimiter{\abs}{\lvert}{\rvert}

\newcommand{\thetitle}{Efficient parameterized approximation}
\date{January 24, 2025}

\title{\thetitle} 

\author{Stefan Kratsch}
\author{Pascal Kunz\thanks{Supported by the DFG Research Training Group 2434 ``Facets of Complexity''.}}

\affil{\small Humboldt-Universit\"at zu Berlin, Algorithm Engineering, Berlin, Germany 
\\ \{kratsch,kunzpasc\}@informatik.hu-berlin.de}

\begin{document}

\maketitle
\begin{abstract}
Many important problems are \NP-hard and, unless $\classP=\NP$, they do not admit polynomial-time exact algorithms. In fact, the fastest known algorithms for solving them exactly usually take time exponential in the input size. Much research effort has gone into obtaining much faster exact algorithms for instances that are sufficiently well-structured, e.g., through parameterized algorithms with running time $f(k)\cdot n^{\Oh(1)}$ where $n$ is the input size and $k$ quantifies some structural property such as treewidth. When $k$ is small, this is comparable to a polynomial-time exact algorithm and usually it outperforms the fastest exact exponential-time algorithms for a large range of $k$.

In this work, we are interested instead in leveraging instance structure for polynomial-time approximation algorithms. Concretely, we aim for polynomial-time algorithms that produce a solution of value at most or at least (depending on minimization vs.\ maximization) $c \OPT \pm f(k)$ where $c$ is a constant. Unlike for standard parameterized algorithms, we do not assume that any structural information is provided with the input. Ideally, we can obtain algorithms with small additive error only, i.e., $c=1$ and $f(k)$ is polynomial or even linear in $k$. For small $k$, this is similarly comparable to a polynomial-time exact algorithm and it will beat general case approximation for a large range of $k$.

We study \vertexcover, \convertexcover, \chromaticnumber, and \trianglepacking. The parameters we consider are the size of minimum modulators to graph classes on which the respective problem is tractable. For most problem--parameter combinations that we consider we give algorithms that compute a solution of size at least or at most $\OPT \pm k$. In the case of \textsc{Vertex Cover}, most of our algorithms are tight under the Unique Games Conjecture and provide considerably better approximation guarantees than standard $2$-approximations whenever the modulator is smaller than the optimum solution.
\end{abstract}

\section{Introduction}
	\label{sec:intro}

	Many important (graph) problems are \NP-hard, ruling out polynomial-time exact algorithms that optimally solve every input, unless $\classP=\NP$. Taking into account only the \emph{size} of inputs, the best general (i.e., worst-case) results one can thus hope for are polynomial-time approximation algorithms and super-polynomial-time exact algorithms (usually leading to exact exponential-time algorithms). At the same time, it is well known that hard problems usually admit efficient algorithms on sufficiently well-structured inputs. How much can the \emph{structure} of inputs help us to improve upon the best general-case approximation or exact algorithms?

	Regarding the influence of structure on the \emph{time complexity} of (mainly) exact algorithms, this is the main focus of the field of \emph{parameterized complexity}. Therein, structural properties are quantified by \emph{parameters} (such as solution size, similarity to a tractable special case, or treewidth) and one seeks algorithms that run in time $f(k)\cdot|x|^{\Oh(1)}$ for inputs $(x,k)$ with parameter value $k$, called \emph{fixed-parameter tractable (FPT)} algorithms. For small values of the parameter $k$ this is almost as good as an efficient algorithm and for a certain range of $k$ it will outperform the best known exact exponential-time algorithm. Such algorithms are known for a wide variety of problems with an (equal) variety of parameters.

	There has been far less work on leveraging structure to design polynomial-time approximation algorithms.\footnote{Notably, the active subfield of \emph{parameterized approximation algorithms} studies approximation algorithms that are allowed to run in FPT-time $f(k)\cdot|x|^{\Oh(1)}$, especially for problems with no exact FPT-algorithm, but this does not address polynomial-time approximation.} The natural objective here is to study the influence of structure on the \emph{approximation ratio} of polynomial-time algorithms. Mimicking the benefits of FPT-algorithms, i.e., closeness to a polynomial-time exact algorithm for small parameter value, leads us to \emph{additive approximation} as the main goal: In polynomial time an \emph{efficient parameterized algorithm} with parameter $k$ should return a solution of cost at most $\OPT+f(k)$ (resp., value at least $\OPT-f(k)$) for a minimization (resp.\ maximization) problem; for small values of $k$ this is close to an exact solution. Clearly, it is desirable for $f(k)$ to be no more than polynomial in $k$, or even linear in $k$, such as to beat the best unparameterized (i.e., standard) approximation algorithm for a large range of $k$. 
	
	Intuitively, for many standard problems we should not expect an additive error of at most $\mathrm{poly}(k)$ for parameters that do not grow under disjoint union of instances, such as treewidth: As a toy example consider \vertexcover with parameter $k$ being the largest size of a connected component. By taking a disjoint union of $\mathrm{poly}(n)+1$ copies of an instance with $n$ vertices, the overall solution with error at most $\mathrm{poly}(k)\leq \mathrm{poly}(n)$ would have to contain an optimal solution for at least one copy (yielding $\classP=\NP$).\footnote{\chromaticnumber is a notable exception to this observation. Indeed, it is well known that any graph can in polynomial time be colored with at most $k+1$ colors where $k$ is its degeneracy, which subsumes e.g.\ treewidth at most $k$.} At the same time, cost at most $\OPT+2^k$ is straightforward to get for this example: It holds trivially when $n\leq 2^k$; and for $n>2^k$ we can exactly solve each connected component and obtain an optimal solution in polynomial time. This is clearly reminiscent of what is known about the existence of \emph{polynomial kernelizations} (a parameterized form of efficient preprocessing), which mostly works only for parameters like solution size or distance to a tractable case (i.e., size of a suitable modulator); we will return to this relationship later. It also emphasizes once more that we should mainly be concerned with polynomial or even linear additive error (and will hence focus on modulator-based parameters).
	
	Among the few works on efficient parameterized approximation algorithms, the \emph{structural rounding} framework of Demaine et al.~\cite{Demaine2019} stands out as being closest to our present research interest. They also focus on modulator-based parameters, i.e., on the distance of inputs to some known tractable case of the problem, and show a general approach for getting better than general $\alpha\cdot\OPT$ approximation results for instances that are sufficiently close to a specific tractable case (i.e, with sufficiently small parameter value). This involves a sufficiently good approximation of the distance, efficiently solving the remaining tractable case (with a known algorithm), and lifting back to an approximate solution of the initial input (with cost/value change proportional to the approximate distance). Here the final step implicitly produces an additive approximation but it is traded for $\alpha\cdot\OPT$ using the assumption that the distance is not too large.\footnote{Their framework applies as well for distance to a (better than general) $\alpha$-approximable case and will then allow lift to a slightly worse than that $\alpha\cdot\OPT$ (but better than general) approximation.} Lavallee et al.~\cite{Lavallee2020} evaluated the performance of their approach applied to \vertexcover on practical inputs and found that, empirically, it can outperform the standard $2$-approximation algorithms.
	
	\subparagraph{Our results.}
	Like Demaine et al.~\cite{Demaine2019}, we focus on graph problems parameterized by the vertex-deletion distance to a known tractable graph class $\cC$ for the problem in question, i.e., the size of a \emph{modulator} $X$ such that $G-X\in\cC$. We similarly do not assume that such a modulator is given with the input, though this is commonly assumed to be given for FPT-algorithms and kernelization. 
	In this setting, however, we focus mainly on getting cost at most $\OPT+\alpha k$ (resp.\ value at least $\OPT-\alpha k$) for $\alpha>0$ as low as possible, rather than aiming for $c \cdot \OPT$ (resp.\ $\frac1c \cdot \OPT$) for $c>1$ better than general-case approximation. In many cases, we improve over a straightforward application of structural rounding by not computing an (approximate) modulator but instead seeking a direct argument for why the generated solution fits the claimed error bound.
	
	We will discuss the relationship between efficient parameterized approximation on the one hand and FPT algorithms and structural rounding on the other.
	We then present efficient parameterized approximation algorithms for concrete problems.
	We focus mainly on \vertexcover, \chromaticnumber, and \trianglepacking and obtain the results listed in \cref{tab:results}. (See Section~\ref{section:preliminaries} for definitions of the used parameters/modulators.)
	\begin{table}[t]\centering
		\begin{tabular}{ c | c | c | c}
			Problem & Parameter is modulator to & Guarantee & Ref. \\
			\hline
			\multirow{7}{*}{
			\begin{tabular}{@{}c@{}}\textsc{Weighted}  \\  \textsc{Vertex Cover} \end{tabular}} & cograph & $\OPTVC + 2k$ & \Cref{cor:vc-cograph} \\
			& interval & $\OPTVC + 4k$ & \Cref{cor:vc-interval}\\
			& cluster & $\OPTVC + k$ & \Cref{thm:vc-cvd}\\
			& cocluster & $\OPTVC + k$ & \Cref{cor:vc-ccvd}\\
			& forest & $\OPTVC + k$ & \Cref{thm:vc-fvs}\\
			& chordal & $\frac{3}{2}\OPTVC + k$ & \Cref{cor:vc-chordal}\\
			& bipartite & no $(2-\varepsilon)\OPTVC + c k$ & \Cref{cor:vc-oct}\\
			\hline
			\textsc{Vertex Cover} & split & $\OPTVC + k$ & \Cref{thm:vc-split}\\
			\hline
			\begin{tabular}{@{}c@{}} \textsc{Connected} \\  \textsc{Vertex Cover} \end{tabular} & split & $\OPTVC + k$ & \Cref{thm:cvc-split}\\
			\hline
			\multirow{6}{*}{\chromaticnumber} & bipartite & $2 + k$ & \Cref{cor:col-bounded-c}\\
			& planar & $4 + k$ & \Cref{cor:col-bounded-c}\\
			& chordal & $\OPTCOL + k$ & \Cref{thm:col-chordal}\\
			& cograph & $\OPTCOL + k$ & \Cref{thm:col-cograph}\\
			& cochordal & $2\OPTCOL + k -1$ & \Cref{thm:col-cochordal}\\
			& $(P_3+K_1)$-free & $\OPTCOL + k$ & \Cref{thm:col-p3k1free}\\
			\hline
			\multirow{2}{*}{\textsc{Triangle Packing}} & cluster & $\OPTTP - k$ & \Cref{thm:tp-cvd}\\
			& cocluster & $\OPTTP - k$ & \Cref{thm:tp-ccvd}\\
			\hline
		\end{tabular}
		\caption{An overview of our results: The third column lists the solution quality guaranteed by the algorithm we give.
		Here, $k$ refers to the parameter value and $c$ to an arbitrary constant. Throughout, the input consists of just a graph with no further information about its structure.
		For \textsc{Weighted Vertex Cover}, where the input is a weighted graph, the parameter value is the minimum weight of a modulator to the respective graph class whereas for the other problems it is the minimum size of a modulator.}
		\label{tab:results}
	\end{table}

	\subparagraph{Related work.}
	Efficient parameterized approximation with additive error at most polynomial in some structural parameter attains close to optimal solutions when the parameter value is small, and intuitively this will mostly be restricted to modulator-type parameters. An alternative algorithmic approach for this type of structure is to apply FPT-algorithms for the same parameter, which will return an exact solution though with running time depending exponentially on the parameter value. Depending on the specific setting of problem and parameter choice, efficient parameterized approximation may be a sound alternative for making good use of such structure.
	
	There a few other works that also make use of structure in the input for polynomial-time approximation.
	For example, Chalermsook et al.~\cite{Chalermsook2023} give an $\bigO(\mathrm{tw}\cdot(\log \log \mathrm{tw})^3/\log^3\mathrm{tw})$)-factor approximation for \textsc{Independent Set} where $\mathrm{tw}$ refers to treewidth.
	In contrast to our results, this approximation guarantee is multiplicative rather than additive.
	In the spirit of ``above guarantee'' parameterization, Mishra et al.~\cite{Mishra2011} showed that there is a polynomial-time algorithm that outputs a vertex cover of size at most $\OPT_{\mathrm{MM}} + \bigO(\log n \log \log n)(\OPT_{\mathrm{VC}} \allowbreak - \allowbreak \OPT_{\mathrm{MM}})$ where $\OPT_{\mathrm{MM}}$ denotes the size of a maximum matching and $n$ the number of vertices.
	Fomin et al.~\cite{Fomin2023} obtained similar results for the problem \textsc{Longest Cycle}.
	Kashaev and Sch\"afer~\cite{Kashaev2023} investigate \textsc{Vertex Cover} approximation on graphs that are close to being bipartite.
	They measure this proximity, not in terms of the modulator size, but in terms of the odd girth, that is the length of a shortest odd cycle, but assume that a small modulator to a bipartite graph is given or can be approximated efficiently.	
	Several standard approximation algorithms could also be interpreted as efficient parameterized approximations.
	For example, the standard $2$-approximation for \textsc{Vertex Cover} actually computes a vertex cover of size at most $\OPTVC+\OPT_\mathrm{MM}$.
	In the same vein, the aforementioned simple greedy coloring algorithm achieves a $(k +1)$-coloring, $k$ being the degeneracy of the input graph.
	
	In the field of exponential-time parameterized approximation, the work of Inamdar et al.~\cite{Inamdar2023} is most similar to ours.
	It, too, focuses on ``distance-to-triviality'' parameters, i.e., parameters that measure a graph's distance to a graph class on which a given problem can be solved efficiently.
	
	\subparagraph{Organization.}
	We give basic definitions in \cref{section:preliminaries}.
	In \cref{sec:epa}, we give a formal definition of efficient parameterized approximation and compare this approach to both FPT algorithms and to the structural rounding paradigm.
	In \cref{sec:vc,sec:cvc,sec:col,sec:tp}, we give concrete examples of efficient parameterized approximation algorithms for the problems \textsc{(Weighted) Vertex Cover}, \textsc{Connected Vertex Cover}, \chromaticnumber, and \textsc{Triangle Packing}, respectively.
	We conclude in \cref{section:conclusion}, by pointing out avenues for further research.

\section{Preliminaries}\label{section:preliminaries}
	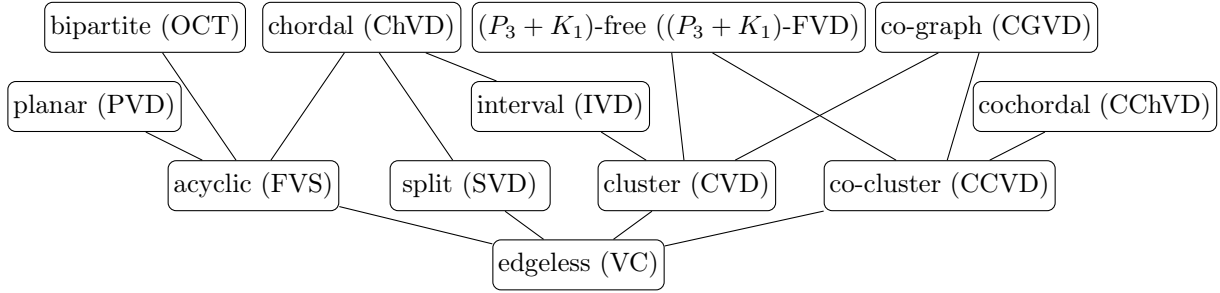
\begin{figure}
		\centering
		\begin{tikzpicture}[yscale=0.7,xscale=0.95]
		\tikzset{
			param/.style={draw, fill=white, rectangle, rounded corners=3, minimum width=5.5em, minimum height=4ex,font={\fontsize{10pt}{12}\selectfont},inner sep=2pt,inner ysep=+0pt
		},
		}
		\begin{scope}[yscale=-1,xscale=1]
		\node[param] at (-7, 1.5) (vc) {edgeless (VC)};
		\node[param] at (-11.5, 0) (fv) {acyclic (FVS)};
		\node[param] at (-2, 0) (ccvd) {co-cluster (CCVD)};
		\node[param] at (-1.35, -3) (cgvd) {co-graph (CGVD)};
		\node[param] at (-5.5, 0) (cvd) {cluster (CVD)};
		\node[param] at (-13, -3) (oct) {bipartite (OCT)};
		\node[param] at (-7.25, -1.5) (pid) {interval (IVD)};
		\node[param] at (-8.5, 0) (split) {split (SVD)};
		\node[param] at (-10, -3) (chordal) {chordal (ChVD)};
		\node[param] at (0.15, -1.5) (cochordal) {cochordal (CChVD)};
		\node[param] at (-5.75, -3) (p3k1free) {$(P_3+K_1)$-free ($(P_3+K_1)$-FVD)};
		\node[param] at (-13.7, -1.5) (planar) {planar (PVD)};
		\end{scope}

		\draw (vc) -- (fv);
		\draw (vc) -- (ccvd);
		\draw (vc) -- (cvd);
		\draw (cvd) -- (pid);
		\draw (fv) -- (oct);
		\draw (cvd) -- (cgvd);
		\draw (ccvd) -- (cgvd);
		\draw (vc) -- (split);
		\draw (pid) -- (chordal);
		\draw (split) -- (chordal);
		\draw (p3k1free) -- (cvd);
		\draw (p3k1free) -- (ccvd);
		\draw (fv) -- (chordal);
		\draw (cochordal)  -- (ccvd);
		\draw (planar) -- (fv);
		\end{tikzpicture}
		\caption{A Hasse diagram of the graph classes and related modulator parameters used in this work:
		A line from one graph class to another means that the lower of the two classes is contained in the other.
		We will use abbreviations to refer to the vertex deletion problem to these graph classes.
		These abbreviations are indicated in parentheses.}
		\label{fig:param-hier}
	\end{figure}

	For any set $X$ and $k\in \mathbb N$, we let $\binom{X}{k} \coloneqq \{Y \subseteq X \mid \abs{Y}= k \}$.
	For any $f\colon X \to \Qpos$ and $Y \subseteq X$, we denote $f(Y) \coloneqq \sum_{x\in Y} f(x)$. 

	Let $G=(V,E)$ be a graph.
	We denote the \emph{open} and \emph{closed neighborhood} of a vertex $v\in V$ by $N_G(v)\coloneqq \{u \in V \mid \{u,v\} \in E\}$ and $N_G[v] \coloneqq N(v) \cup \{v\}$, respectively.
	Two vertices $u,v \in V$ are \emph{true twins} if $N_G[u] = N_G[v]$ and \emph{false twins} if $N_G(u) = N_G(v)$.
	The complement of $G$ is denoted by $\overline{G} \coloneqq (V,\{\{u,v\} \in \binom{V}{2} \mid \{u,v\} \notin E \})$.
	We will denote the size of a maximum clique in a graph $G$ by $\omega(G)$ and the size of a maximum independent set by $\alpha(G)$.
	
	We will now define the four problems that most of this work is concerned with.
	A vertex set $X\subseteq V$ is a \emph{(connected) vertex cover} if $G-X$ is edgeless (and $G[X]$ is connected).
	\boxproblem{(Connected) Vertex Cover ((C)VC)}{A graph $G=(V,E)$.}{A (connected) vertex cover $X$.}{Minimize $\abs{X}$.}
	We will also consider \textsc{Weighted Vertex Cover} (WVC) where the input additionally contains a weight function $w\colon V \to \Qpos$ and the objective is to minimize $w(X)$ rather than $\abs{X}$.
	
	A function $c \colon V \to \mathbb N$ is a \emph{coloring} if $c(u)\ne c(v)$ for all $\{u,v\} \in E$.
	\boxproblem{\chromaticnumber (COL)}{A graph $G=(V,E)$.}{A coloring $c$.}{Minimize the number of used colors, i.e., $\abs{\{i \in \mathbb N \mid c^{-1}(i) \ne \emptyset\}}$.}
	A \emph{triangle packing} in $G$ is a set $\calT$ such that $G[T]$ induces a triangle for every $T\in \calT$.
	\boxproblem{Triangle Packing (TP)}{A graph $G=(V,E)$.}{A triangle packing $\calT$.}{Maximize $\abs{\calT}$.}
	
	We will now define the parameters used in this work.
	If $\calC$ is a graph class and $G=(V,E)$ an arbitrary graph, then $M\subseteq V$ is a \emph{modulator} to $\calC$ in $G$ if $G-M \in \calC$.
	It is a \emph{minimum modulator} to $\calC$ if $\abs{M'} \ge \abs{M}$ for all modulators $M'$ in $G$.
	If $w\colon V\to \Qpos$ is a weight function on $G$, then $M$ is a \emph{minimum-weight modulator} to $\calC$ if $w(M') \ge w(M)$ for all modulators $M'$ in $G$.
	
	Most of the parameters we will consider will be the minimum size or weight of a modulator in the given graph to a certain graph class.
	We will now define the relevant graph classes.
	A diagram describing the containment relationships between these classes is given in \Cref{fig:param-hier}.
	As noted above, a modulator to an edgeless graph is a \emph{vertex cover}.
	A modulator to an acyclic graph is a \emph{feedback vertex set}.
	The graph $G=(V,E)$ is \emph{bipartite} if it does not contain an odd cycle and a modulator to bipartite graph is called an \emph{odd cycle transversal}.
	It is a \emph{split graph} if $V = I \uplus C$ where $G[I]$ is edgeless and $G[C]$ is complete.
	It is \emph{chordal} if it does not contain an induced cycle on at least four vertices.
	It is \emph{cochordal} if $\overline{G}$ is chordal.
	It is a \emph{cluster graph} if it does not contain a path on three vertices as an induced subgraph.
	It is a \emph{cocluster} if $\overline{G}$ is a cluster graph.
	It is a \emph{cograph} if it does not contain a path on four vertices as an induced subgraph.
	It is an \emph{interval graph} if there is a set $\calI$ of intervals on the real line and a bijection $\tau \colon V \to \calI$ such that $\{u,v\} \in E$ if and only if $\tau(u)$ and $\tau(v)$ overlap.
	A graph is \emph{$(P_3+K_1)$}-free if it does not contain the graph $P_3+K_1$ (a drawing of which is given in \Cref{fig:p3k1}) as an induced subgraph.

\section{Efficient parameterized approximation}
	\label{sec:epa}
	An \emph{optimization problem} over an alphabet $\Sigma$ is defined by a function $\calP \colon \Sigma^* \times \Sigma^* \to \mathbb R \cup \{\pm \infty\}$.
	An optimization problem can be a minimization or a maximization problem.
	Intuitively, $\calP(I,X)$ is the value of a solution $X$ in an instance $I$.
	For each $I\in \Sigma^*$, the number of $X\in \Sigma^*$ such that $\calP(I,X) \notin \{\pm \infty\}$ should be finite.
	Then, we define the optimum value of $\calP$ on $I$ as $\OPT_\calP(I) \coloneqq \min_{X \in \Sigma^*} \calP(I,X)$ if $\calP$ is a minimization problem or as $\OPT_\calP(I) \coloneqq \max_{X \in \Sigma^*} \calP(I,X)$ if it is a maximization problem.
	A \emph{constant-factor approximation} with ratio $\alpha \ge 1$ for $\calP$ is a polynomial-time algorithm that, given an instance $I \in \Sigma^*$, outputs a solution $X\in \Sigma^*$ such that $\calP(I,X) \le \alpha \OPT_\calP(I)$ if $\calP$ is a minimization problem or $\calP(I,X) \ge \frac{1}{\alpha} \OPT_\calP(I)$ if it is a maximization problem.
	
	A \emph{parameterized optimization problem} is a pair consisting of an optimization problem $\calP$ and a function $\kappa \colon \Sigma^* \to \mathbb N$.
	Let $\alpha \ge 1$ and $f\colon \mathbb N \to \mathbb R_{\ge 0}$.
	An \emph{efficient parameterized approximation} (EPA) with an additive error of $f$ and a ratio of $\alpha$ is a polynomial-time algorithm that given an instance $I\in \Sigma^*$ outputs a solution $X\in \Sigma^*$ with $\calP(I,X) \le \alpha \OPT_\calP +  f(\kappa(I))$ if $\calP$ is a minimization problem or $\calP(I,X) \ge \frac{1}{\alpha} \OPT_\calP - f(\kappa(I))$ if it is a maximization problem.
	It is important to note that an EPA is not given access to the parameter value $\kappa(I)$ or any structure witnessing the parameter value.
	In fact, designing EPAs is trivial in many cases when such a structure is given.
	While in parameterized complexity it is fairly common to consider the solution size as a parameter (it is sometimes called the \emph{standard parameterization}), this is not interesting in the context of efficient parameterized approximation, since this is equivalent to normal polynomial-time approximation.
	Consider, for instance, a minimization problem with an EPA with ratio $\alpha$ and additive error $f(k)$ where the parameter $k$ is the optimum solution size.
	Then, this EPA is simply a polynomial-time approximation with factor $\alpha + f(\OPT)/\OPT$.
	On the other hand, graph parameters that do not increase when a graph is copied, while popular in parameterized algorithms, are also of little use if one wants an EPA with polynomial additive error for most common graph problems, as we noted in the introduction.
	This is why, in the present work, we focus almost exclusively on parameterizations by modulator parameters.

	On the other hand, there are trivial EPAs with superpolynomial additive error for such parameterizations.
	We will say that $(\calP,\kappa)$ is \emph{fixed-parameter tractable}, if there is a computable function $f\colon \mathbb N \to \mathbb N$ and an algorithm with running time bounded by $f(\kappa(I)) \cdot \abs{I}^{\bigO(1)}$ that given an instance $I \in \Sigma^*$ outputs an optimum solution, that is, a solution $X\in \Sigma^*$ with $\calP(I,X) = \OPT_\calP(I)$.
	This is a somewhat non-standard definition in parameterized complexity, as in this field one generally assumes that an FPT algorithm is given the parameter value and sometimes even a structure witnessing this parameter value as input and also because when dealing with exact parameterized algorithms one usually considers only decision problems.
	In practice, neither of these deviations from the normal definitions should make much of a difference for typical problems, parameters, or algorithms.

	We will first show that, if a parameterized optimization problem that satisfies certain mild conditions admits an FPT algorithm, then it also admits an EPA with a ratio of $1$ and an additive error which depends on the superpolynomial part of the running time of the FPT algorithm.
	This is similar to the relationship between FPT algorithms and kernelization where it is known (see \cite{Cygan2015}, for instance) that the existence of an FPT algorithm implies a kernelization, albeit one with superpolynomial size.
	We need the additional assumption that for some constant $c$ there is a polynomial-time algorithm that outputs a solution $X$ with $\calP(I,X) \le \OPT_\calP(I) +  \abs{I}^c$ if $\calP$ is a minimization problem
	or $\calP(I,X) \ge \OPT_\calP(I) -  \abs{I}^c$ if it is a maximization problem.
	We note that this condition is trivial for most typical optimization problems (for instance, in the case of \textsc{Vertex Cover}, one may output the entire graph, or, in the case of \textsc{Independent Set}, the empty set).
	
	\begin{theorem}
		\label{thm:FPT-EPA}
		Let $(\calP,\kappa)$ be a parameterized minimization (resp. maximization) problem such that there is a polynomial-time algorithm that outputs a solution $X$ with $\calP(I,X) \le \OPT_\calP(I) +  \abs{I}^c$ (resp. $\calP(I,X) \ge \OPT_\calP(I) -  \abs{I}^c$) for some constant $c$.
		If $(\calP,\kappa)$ is FPT with a running time of $f(k)\cdot \abs{I} ^{\bigO(1)}$, then $(\calP,\kappa)$ admits an EPA with a ratio of $1$ and additive error of $f$.
	\end{theorem}
	\begin{proof}
		Let $f(k) \cdot \abs{I}^d$ with for a constant $d$ be the running time of the FPT algorithm for $(\calP,\kappa$).
		The EPA for $(\calP,\kappa)$ works as follows on input $I$.
		It first runs the FPT algorithm for $\abs{I}^{c+d}$ steps.
		If the FPT algorithm terminates within that number of steps, then it outputs an optimum solution, so the EPA outputs this solution.
		If this algorithm does not terminate, then the EPA computes a solution $X$ with $\calP(I,X) \le \OPT_\calP(I) +  \abs{I}^c$ (resp. $\calP(I,X) \ge \OPT_\calP(I) -  \abs{I}^c$) and outputs this solution.
		In the latter case, we know that $f(k) \cdot \abs{I}^d \ge \abs{I}^{c+d}$ and, therefore, $f(k) \ge \abs{I} ^ c$.
	\end{proof}

	Unlike in the analogous statement for kernelization that we pointed to, the converse of \cref{thm:FPT-EPA} is not true.
	Consider, for example, the \textsc{Edge Coloring} problem, in which one is given a graph and is asked to assign colors to the edges of the graph such that no two edges that share a vertex are assigned the same color.
	The goal is to minimize the number of colors used.
	Testing whether a graph can be edge colored with three colors is NP-hard~\cite{Holyer1981}, so, unless $\mathrm{P}=\mathrm{NP}$, there is no FPT algorithm for \textsc{Edge Coloring} parameterized by the number of colors.
	There is, however, an EPA with ratio and additive error $1$ for this problem because Vizing's algorithm~\cite{Misra1992} only uses $\OPT + 1$ colors.
	
	On the other hand, there are also problems that admit a polynomial (exact) kernel, but do not admit an EPA with ratio~$1$ and a polynomial additive error.
	Consider the \textsc{Tree Deletion Set} problem, where the task is to find a minimum modulator to a tree, parameterized by the solution size.
	This problem admits a polynomial kernel, but no $\OPT^{\bigO(1)}$-factor approximation, unless $\mathrm{P}=\mathrm{NP}$~\cite{Giannopoulou2016}.
	As we noted above, an EPA with polynomial additive error would be equivalent to such an approximation.

	As we noted in the introduction, the present work is closely related to the ``structural rounding'' framework introduced by Demaine et al.~\cite{Demaine2019}.
	Let $\Psi$ be a class of edit operations on graphs, that is, each element of $\Psi$ is a function that inputs and outputs a graph.
	Such a class may include vertex deletions, edge deletions, edge additions, etc.
	A graph $G'$ is \emph{$\gamma$-editable} from a graph $G$ under $\Psi$, if there is a sequence of edits $\psi_1,\ldots,\psi_k \in \Psi$ with $k\le \gamma$ such that $G' = \psi_k(\psi_{k-1}(\cdots\psi_1(G) \cdots ))$.
	The graph $G$ is \emph{$\gamma$-close} to the graph class $\calC$, if there is a graph $G'\in \calC$ such that $G'$ is $\gamma$-editable from $G$ under $\Psi$.
	A graph minimization (resp. maximization) problem $\calP$ is \emph{stable} under $\Psi$ with constant $c'$ if $\OPT_\calP(G') \le \OPT_\calP(G) + c'\gamma$ (resp. $\OPT_\calP(G') \le \OPT_\calP(G) - c'\gamma$) for any graph $G'$ that is $\gamma$-editable from $G$ under $\Psi$.
	A minimization (resp. maximization) problem $\calP$ can be \emph{structurally lifted} with respect to $\Psi$ with constant $c$ if there is a polynomial-time algorithm that given graphs $G,G'$, an edit sequence such that $G' = \psi_k(\psi_{k-1}(\cdots\psi_1(G) \cdots ))$, and a solution $S'$ for $G'$ and outputs a solution $S$ for $G$ such that $\calP(G,S) \le \calP(G',S') + ck$ (resp. $\calP(G,S) \ge \calP(G',S') - ck$).	
	Let $\calQ(\calC,\Psi)$ denote the problem of finding a minimum number of edits from $\Psi$ in an input graph to obtain a graph in $\calC$.
	
	The main theorem of Demaine et al.~\cite{Demaine2019} can be restated as guaranteeing the existence of an EPA.
	For completeness and because our statement of the theorem differs from that of Demaine et al., we include a proof.
	\begin{theorem}[Demaine et al.~\cite{Demaine2019}]
		\label{thm:struct-rounding}
		Let $\calP$ be a minimization problem (resp. maximization problem) that is stable under $\Psi$ with constant $c'$ and can structurally lifted with respect to $\Psi$ with constant~$c$.
		Suppose that $\calP$ admits an $\alpha$-approximation algorithm on the graph class $\calC$ and that $\calQ(\calC,\Psi)$ admits a $\beta$-approximation algorithm on general graphs.
		Then, $\calP$ parameterized by the closeness~$k$ of a graph to $\calC$ under $\Psi$ admits an EPA with ratio $\alpha$ and additive error of $(\alpha \beta c' + c\beta)k$ (resp. $(\frac{1}{\alpha} \beta c' + c\beta)k$).
	\end{theorem}
	\begin{proof}
		Let $G$ be a graph that is $k$-close to $\calC$ under $\Psi$.
		Using the $\beta$-approximation for $\calQ(\calC,\Psi)$, our algorithm can compute $G'$, a graph that is $\beta k$-editable from $G$ under $\Psi$ along with the edits $\psi_1,\ldots,\psi_{\beta k}$.
		We then apply the $\alpha$-approximation for $\calP$ on $\calC$ to $G'$ to obtain a solution $X'$.
		We use the structural lifting algorithm to obtain a solution $X$ for $G$.
		
		Consider the case where $\calP$ is a minimization problem.
		Because of the stability of $\calP$ with constant $c'$ and the fact that $G'$ is $\beta k$-editable from $G$,
		\begin{align*}
			\OPT_\calP(G') \le \OPT_\calP(G) + \beta k c'.
		\end{align*}
		Therefore,
		\begin{align*}
			\calP(G',X') \le \alpha(\OPT_\calP(G) + \beta k c')
		\end{align*}
		Because the structural lifting algorithm we apply has constant $c$ and the number of edits is $\beta k $,
		\begin{align*}
			\calP(G,X) \le \calP(G',X') + c \beta k \le \alpha(\OPT_\calP(G) + \beta k c') + c \beta k = \alpha \OPT_\calP(G) + (\alpha \beta c' + c\beta)k.
		\end{align*}
		If $\calP$ is a maximization problem, the proof is analogous.
	\end{proof}

	As an example, consider the \textsc{Vertex Cover} problem parameterized by the size of  a minimum feedback vertex set in the input graph.
	In other words, $\Psi$ contains vertex deletion operations and $\calC$ is the class of acyclic graphs.
	Then, any feedback vertex set of size $k$ corresponds to a sequence of $k$ edits from $\Psi$ to $\calC$.
	It is easy to see that \textsc{Vertex Cover} is stable under $\Psi$ with constant $c'=0$ and that it can be structurally lifted with constant $c=1$ by simply adding the entire feedback vertex set to the solution.
	Moreover, there is a $1$-approximation for \textsc{Vertex Cover} on acyclic graphs and $\calQ(\calC,\Psi)$, i.e. the \textsc{Feedback Vertex Set} problem, has a $2$-approximation~\cite{Bafna1999,Becker1996}.
	Therefore, \cref{thm:struct-rounding} implies an EPA with ratio $1$ and additive error of $2k$.
	In \cref{sec:vc-fvs}, we will improve on this by giving an EPA with ratio $1$ and an additive error of $k$.

\section{Vertex Cover}
\label{sec:vc}
		In this section, we will consider the problems \textsc{Vertex Cover} and \textsc{Weighted Vertex Cover}.
		
		We will give several efficient parameterized approximations for these two problems parameterized by the size of a minimum modulator to certain graph classes on which they can be solved in polynomial time.
		Most of these EPAs have a ratio of $1$.
		We briefly observe that such an EPA is at the same an EPA for the \textsc{Independent Set} (IS) problem with the same parameter.
		
		\begin{observation}
			\label{obs:vc-is}
			If \textsc{(Weighted) Vertex Cover} with parameter $\kappa$ has an EPA with ratio $1$ and additive error $f$, then \textsc{(Weighted) Independent Set} with parameter $\kappa$ has an EPA with the same ratio and the same additive error.
		\end{observation}
		\begin{proof}
			We will prove the claim in the weighted case, since the unweighted case follows by taking a uniformly weighted graph.
			Note that for any weighted graph $G=(V,E)$ with weight function $w$, $\OPT_\mathrm{IS}(G) = w(V) - \OPTVC(G)$ and that the complement of any vertex cover is an independent set.
			Hence, if we output the complement of the result computed by an EPA for \textsc{Weighted Vertex Cover} with ratio $1$ and additive error $f$, we get an independent set of weight at most $w(V) - \OPTVC(G) - f(\kappa) = \OPT_\mathrm{IS}(G) - f(\kappa)$.
			Hence, we get an EPA for \textsc{Independent Set} with ratio $1$ and additive error $f$.
		\end{proof}
		
		Note that \cref{obs:vc-is} does not apply to EPAs with ratio greater than $1$.
		
		Several of our algorithms will use the well-known local ratio method.
		It is based on the following is an adaptation of what is known as the \emph{Local Ratio Theorem} or \emph{Local Ratio Lemma} (cf. e.g. \cite{BarYehuda2004}).
		For simplicity, we only state and prove the lemma for \textsc{Vertex Cover}, though it certainly applies to a much wider range of problems.
		
		\begin{lemma}
			\label{lemma:local-ratio}
			Let $\calC$ be a graph class and $\calP$ be the problem of finding a minimum-weight modulator to $\calC$. Let $\alpha,\beta \ge 1$.
			Let $G=(V,E)$ be a graph and $w,w_1,w_2 \colon V \to \Qpos$ be weight functions with $w = w_1 + w_2$.
			Let $X$ be a vertex cover in $G$
			such that $w(X) \le \alpha \OPTWVC(G,w_i) + \beta \OPT_\calP(G,w_i)$ for $i\in\{1,2\}$.
			Then, $w(X) \le \alpha \OPTWVC(G,w) + \beta \OPT_\calP(G,w)$.
		\end{lemma}
		\begin{proof}
			Let $X^*,X_1^*,X_2^*$ be minimum-weight vertex covers in $(G,w)$, $(G,w_1)$, and $(G,w_2)$, respectively.
			Similarly, let $Y^*,Y_1^*,Y_2^*$ be minimum-weight solutions for $\calP$ for $(G,w)$, $(G,w_1)$, and $(G,w_2)$, respectively.
			For $i \in \{1,2\}$, it is clear that $w_i(X_i^*) \le w_i(X^*)$ and  $w_i(Y_i^*) \le w_i(Y^*)$.
			Then,
			\begin{align*}
				w(X) &= w_1(X) + w_2(X) \le \alpha w_1(X_1^*) + \beta w_1(Y_1^*) + \alpha w_2(X_2^*) + \beta w_2(Y_2^*)\\
				&\le \alpha w_1(X^*) + \beta w_1(Y^*) + \alpha w_2(X^*) + \beta w_2(Y^*) \\
				& = \alpha w(X^*) + \beta w(Y^*) = \alpha \OPTWVC(G,w) + \beta \OPT_\calP(G,w). \qedhere
			\end{align*}
		\end{proof}
	
		We use $\alpha(H)$ to denote the independence number of a graph $H$.
		If $\calF$ is a graph class, then we will say that a graph $G$ is \emph{$\calF$-free} if $G$ does not contain any graph in $\calF$ as an induced subgraph.
	
		\begin{theorem}
			\label{thm:F-del}
			Let $\calF$ be a graph class such that
			\begin{compactenum}[(i)]
				\item there is a polynomial-time algorithm that, given a graph $G=(V,E)$, either correctly decides that $G$ is $\calF$-free or outputs $S\subseteq V$ such that $G[S] \in \calF$,
				\item $\alpha^* \coloneqq \max_{H \in \calF} \alpha(H)$ is bounded, and
				\item \textsc{Weighted Vertex Cover} can be solved exactly in polynomial time on $\calF$-free graphs.
			\end{compactenum}
			Let $\calP$ denote the problem of finding a minimum-weight modulator to the class of $\calF$-graphs.
			Then, there is a polynomial-time algorithm that, given a weighted graph $(G,w)$, finds a vertex cover of weight at most $\OPTWVC(G,w) + \alpha^* \OPT_\calP(G,w)$.
		\end{theorem}
		\begin{proof}
			The algorithm proceeds in the following manner given a graph $G=(V,E)$:
			If $G \in \calF$, then it computes a minimum vertex cover using the given algorithm.
			If there is a vertex $v$ with $w(v)=0$, then this algorithm is recursively called on $G-v$ to obtain a solution $X'$ and the algorithm outputs $X\cup \{v\}$.
			Otherwise, $G$ contains $S\subseteq V$ such that $G[S] \in \calF$
			The weights of all vertices in $S$ are reduced by $\lambda^* \coloneqq \min_{v\in S} w(v)$ and this algorithm is recursively called on the graph with the modified weights and the solution for this graph is output.
			
			Clearly, this algorithm can be implemented to run in polynomial time.
			We will show that $w(X) \le \OPTWVC(G,w) + \alpha^* \OPT_\calP(G,w)$ where $X$ is the solution output by the vertex cover by induction on the recursive calls of the algorithm, following the standard pattern for the analysis of the local ratio algorithms.
			If $G \in \calF$, then $w(X) = \OPTWVC(G,w)$.
			If $G$ contains $v\in V$ with $w(v)$, then, by induction, $w(X) = w(X') \le \OPTWVC(G',w) + \alpha^* \OPT_\calP(G',w) \le \OPTWVC(G,w) + \alpha^* \OPT_\calP(G,w)$.
			Now, suppose that $S\subseteq V$ and $G[S] \in \calF$.
			Let
			\begin{align*}
				w_1(v) \coloneqq\begin{cases}
					\lambda^*, & \text{ if } v\in S,\\
					0, & \text{ if } v\notin S,
				\end{cases} \quad \text{ and } \quad 
				w_2(v) \coloneqq\begin{cases}
					w(v) - \lambda^*, & \text{ if } v\in S,\\
					w(v), & \text{ if } v\notin S.
				\end{cases}
			\end{align*}
			Then, $w=w_1+w_2$.
			By induction, $w_2(X) \le \OPTWVC(G,w_2) + \alpha^* \OPT_\calP(G,w_2)$.
			Moreover,
			\begin{align*}
				w_1(X) = \lambda^* \abs{X\cap S} \le \lambda^*\abs{S} =
				\lambda^*(\abs{S}-\alpha^*) + \alpha^*\lambda^*
				\le \OPTWVC(G,w_1) + \alpha^* \OPT_\calP(G,w_1).
			\end{align*}
			Hence, by \cref{lemma:local-ratio}, $w(X) \le \OPTWVC(G,w) + \alpha^* \OPT_\calP(G,w)$.
		\end{proof}
	
		As an example, we consider cographs.
		Recall that a graph is a cograph if and only if it is $P_4$-free. \textsc{Weighted Vertex Cover} can be solved exactly on cographs in polynomial-time.
		There is a naive $4$-approximation for \textsc{Cograph Vertex Deletion} (CGVD), i.e., the problem of finding a minimum-weight modulator to a cograph and no better approximation algorithm is known.
		A naive application of the structural rounding framework, therefore, shows that we can find a vertex cover of weight at most $\OPTWVC + 4\OPTCGVD$ in polynomial-time.
		\cref{thm:F-del} shows that the structural rounding algorithm actually gives a vertex cover of weight at most $\OPTWVC + 2\OPTCGVD$, because $\alpha(P_4) = 2$.
		
		\begin{corollary}
			\label{cor:vc-cograph}
			There is a polynomial-time algorithm that, given a weighted graph $(G,w)$, computes a vertex cover of weight at most $\OPTWVC(G,w) + 2\OPTCGVD(G,w)$.
		\end{corollary}
	
		\begin{figure}
			\centering
			\begin{tikzpicture}[scale=.7]
				\def\nsc{0.5}
				\def\lheight{-.6}
				\tikzstyle{xnode}=[circle,scale=\nsc,draw];
				\begin{scope}[xshift=4cm]
					\node[xnode] (b1) at (1,2) {};
					\node[xnode] (b2) at (0,.5) {};
					\node[xnode] (b3) at (2,.5) {};
					\node[xnode] (b4) at (1,3) {};
					\node[xnode] (b5) at (-.75,-.25) {};
					\node[xnode] (b6) at (2.5,-.25) {};
					\draw (b4) -- (b1) -- (b2) -- (b3) -- (b1);
					\draw (b2) -- (b5);
					\draw (b3) -- (b6);
					\node[] () at (1,\lheight) {net};
				\end{scope}
			
				\begin{scope}[xshift=7.5cm]
					\node[xnode] (c1) at (1,2.5) {};
					\node[xnode] (c2) at (0,1.5) {};
					\node[xnode] (c3) at (2,1.5) {};
					\node[xnode] (c4) at (1,.5) {};
					\node[xnode] (c5) at (-.5,.5) {};
					\node[xnode] (c6) at (2.5,.5) {};
					\draw (c2) -- (c1) -- (c3) -- (c2) -- (c5) -- (c4) -- (c2);
					\draw (c4) -- (c3) -- (c6) -- (c4);
					\node[] () at (1,\lheight) {tent};
				\end{scope}
			
				\begin{scope}[xshift=10.5cm]
					\node[xnode] (d1) at (0,2.5) {};
					\node[xnode] (d2) at (2,2.5) {};
					\node[xnode] (d3) at (2,0.5) {};
					\node[xnode] (d4) at (0,0.5) {};
					\draw (d1) -- (d2) -- (d3) -- (d4) -- (d1);
					\node[] () at (1,\lheight) {$C_4$};
				\end{scope}
			
				\begin{scope}[xshift=13cm]
					\node[xnode] (e1) at (0,2.5) {};
					\node[xnode] (e2) at (1,3) {};
					\node[xnode] (e3) at (2,2.5) {};
					\node[xnode] (e4) at (2,0.5) {};
					\node[xnode] (e5) at (0,0.5) {};
					\draw (e1) -- (e2) -- (e3) -- (e4) -- (e5) -- (e1);
					\node[] () at (1,\lheight) {$C_5$};
				\end{scope}
			
				\begin{scope}[xshift=15.5cm]
					\node[xnode] (a1) at (1,1) {};
					\node[xnode] (a2) at (1,2) {};
					\node[xnode] (a3) at (0,.5) {};
					\node[xnode] (a4) at (2,.5) {};
					\node[xnode] (a5) at (1,1.5) {};
					\node[xnode] (a6) at (0.5,.75) {};
					\node[xnode] (a7) at (1.5,.75) {};
					\draw (a1) -- (a5) -- (a2);
					\draw (a1) -- (a6) -- (a3);
					\draw (a1) -- (a7) -- (a4);
					\node[] () at (1,\lheight) {long claw};
				\end{scope}
			
				\begin{scope}[xshift=18.5cm]
					\node[xnode] (c1) at (1,2.5) {};
					\node[xnode] (c2) at (0,1.5) {};
					\node[xnode] (c3) at (2,1.5) {};
					\node[xnode] (c4) at (0.5,.5) {};
					\node[xnode] (c7) at (1.5,.5) {};
					\node[xnode] (c5) at (-.5,.5) {};
					\node[xnode] (c6) at (2.5,.5) {};
					\draw (c2) -- (c1) -- (c3) -- (c2) -- (c5) -- (c4) -- (c2);
					\draw (c4) -- (c3) -- (c6) -- (c7) -- (c4);
					\draw (c3) -- (c7) -- (c2);
					\node[] () at (1,\lheight) {rising sun};
				\end{scope}
			
				\begin{scope}[xshift=22.5cm]
					\node[xnode] (c1) at (1,1.5) {};
					\node[xnode] (c2) at (0,1.5) {};
					\node[xnode] (c3) at (-1,1.5) {};
					\node[xnode] (c4) at (2,1.5) {};
					\node[xnode] (c5) at (3,1.5) {};
					\node[xnode] (c6) at (1,2.5) {};
					\node[xnode] (c7) at (1,.5) {};
					\draw (c3) -- (c2) -- (c1) -- (c4) -- (c5);
					\draw (c6) -- (c1);
					\draw (c7) -- (c3);
					\draw (c7) -- (c2);
					\draw (c7) -- (c1);
					\draw (c7) -- (c4);
					\draw (c7) -- (c5);
					\node[] () at (1,\lheight) {whipping top};
				\end{scope}
	
			\end{tikzpicture}
			\caption{Small forbidden induced subgraphs for interval graphs.}
			\label{fig:forbidden-int}
		\end{figure}
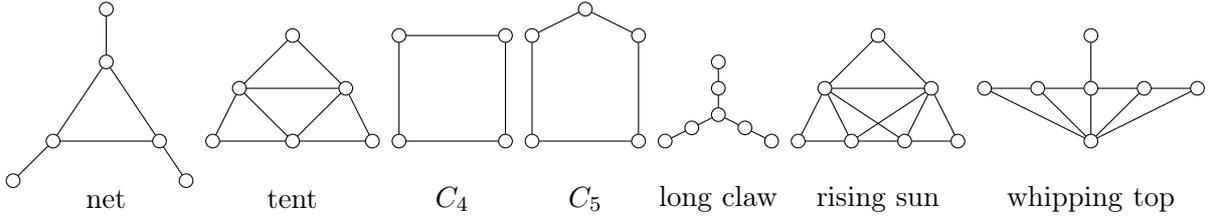
	
		Next, we consider interval graphs.
		\cref{thm:F-del} is not directly applicable for the parameter $\OPTIVD$, the weight of a minimum modulator to an interval graph, because the set of minimal forbidden induced subgraphs for interval graphs is infinite and has unbounded independence number.
		Nevertheless, we are able to find a vertex cover of weight at most $\OPTWVC + 4 \OPTIVD$ by considering a superset of interval graphs that can be characterized by a finite number of forbidden induced subgraphs that each have independence number at most $4$.
		The main challenge is showing that \textsc{Weighted Vertex Cover} can solved exactly in polynomial time on this superset of interval graphs.
		We will make use of arguments in Cao's fixed-parameter tractable algorithm for \textsc{Interval Vertex Deletion}~\cite{Cao2016}.
		We first require some terminology.
		
		Let $G=(V,E)$ be a graph.
		A vertex set $M\subseteq V$ is a \emph{module} if $N(v)\setminus M = N(v')\setminus M$ for all $v,v'\in M$.
		Singleton sets, $V$, and the empty set are \emph{trivial modules}, while all other modules are called \emph{non-trivial}.
		The graph is \emph{prime} if it has no non-trivial modules.
		For any two disjoint modules $M$ and $M'$, either all pairs $v\in M$ and $v'\in M'$ are adjacent or none of them are.
		If $M_1,\ldots,M_p$ constitute a partition of $V$ and each $M_i$ is a module, then this partition is a \emph{modular partition}.
		The graph $Q$ is a \emph{quotient graph} of $G$, if the vertices of $Q$ are in a one-to-one correspondence with the modules of a modular partition of $G$ such that two vertices in $Q$ are adjacent if and only if all pairs of vertices in the corresponding modules are adjacent.
		The following statement is a classic result:
		If both $G$ and its complement are connected, then $G$ has a prime quotient graph~\cite{Gallai1967}.
		A graph $H$ is a \emph{clique decomposition} of $G$ if the vertices of $H$ are in one-to-one correspondence with the maximal cliques of $G$ and, for any vertex $v$ of $G$, the vertices of $H$ corresponding to cliques that contain $v$ induce a connected subgraph of $H$.
		
		Let $\calF=\{\text{net}, \text{tent}, \text{rising sun}, \text{long claw}, \text{whipping top}, C_4, C_5\}$ (see~\Cref{fig:forbidden-int}·).
		By simple inspection we can observe that $\max_{H \in \calF} \alpha(H) = 4$.
		Moreover, it is known~\cite{Lekkeikerker1962} that all graphs in $\calF$ are forbidden for interval graphs, though there more minimal forbidden graphs.
		A \emph{caterpillar} is a tree graph consisting of a single path and possibly leaves attached to some of the vertices in the path.
		An \emph{olive ring} is a graph consisting of a single cycle and possibly leaves attached to some of the vertices in the cycle.
		
		\begin{theorem}[Cao~\cite{Cao2016}]
			\label{thm:Cao}
			A prime $\calF$-free graph $G$ has a clique decomposition that is either a caterpillar or an olive ring.
			This decomposition can be computed in linear time.
		\end{theorem}
	
		This implies a fairly straightforward polynomial-time algorithm for \textsc{Weighted Vertex Cover} on $\calF$-free graphs.
	
		\begin{lemma}
			\textsc{Vertex Cover} can be solved exactly in polynomial time on $\calF$-free graphs.
		\end{lemma}
		\begin{proof}
			The algorithm is based on the following well-known fact:
			For any module $M$ and minimum-weight vertex cover $X$ in a graph $G$, $X$ must contain either all of $M$ or a minimum-weight vertex cover in $G[M]$.
			
			The algorithm then proceeds as follows:
			If the input graph $G$ is disconnected, we recursively solve each connected component separately and output the union of the solutions.
			If the complement $\overline{G}$ of $G$ is disconnected, and $V_1$ and $C_2$ are the vertex sets of the two connected components of $\overline{G}$, we recursively solve each of the $G[V_1]$ and $G[V_2]$ separately, obtaining the solutions $X_1$ and $X_2$, and output $X_j$ with $j = \mathrm{arg\,min}_{i\in\{1,2\}} w(X_i) + w(V_{3-i})$.
			
			Otherwise, $G$ has a prime quotient graph $Q$.
			We compute $Q$ and the corresponding modular partition of $G$ and solve $G[M]$ for each module $M$ of the partition to obtain a solution $X_M$.
			For each vertex $v$ of $Q$, we then let $w(v) \coloneqq w(M) - w(X_M)$ where $M$ is the module corresponding to $M$ and we now consider the weighted graph $(Q,w)$.
			We use \cref{thm:Cao} to compute a clique decomposition.
			If this decomposition is a caterpillar, then $Q$ is chordal and, hence, we can solve \textsc{Weighted Vertex Cover} in linear time on $Q$.
			If the decomposition is an olive ring, then let $C$ be an arbitrary maximal clique of $Q$ located on the cycle of the decomposition.
			Any vertex cover must contain all or all but one of the vertices in $C$.
			Hence, we compute $\abs{C} + 1$ solutions, depending on whether every vertex in $C$ is in the solution and, if not, which one is omitted.
			In each case, we delete all vertices in $C$ and all vertices adjacent to the vertex omitted.
			The remaining graph is again chordal and, therefore, known polynomial-time algorithms can be applied.
		\end{proof}
	
		\begin{corollary}
			\label{cor:vc-interval}
			There is a polynomial-time algorithm that, given a weighted graph $(G,w)$, computes a vertex cover of weight at most $\OPTWVC(G,w) + 4\OPTIVD(G,w)$.
		\end{corollary}
	
		\subsection{Cluster and co-cluster vertex deletion}
		\label{sec:vc-cvd}
		We now consider the problem of computing vertex covers in graphs that are close to being cluster graphs.
		Recall that a graph is a \emph{cluster graph} if every connected component is a clique or, equivalently, if it does not contain $P_3$, the path on three vertices, as an induced subgraph.
		The problem of finding a minimum-weight modulator to a cluster graph is called \textsc{Cluster Vertex Deletion} (CVD).
		We will show that there is a polynomial-time algorithm that outputs a vertex cover of weight at most $\OPTWVC + \OPTCVD$.
		Our algorithm is very similar to the $2$-approximation for \textsc{Cluster Vertex Deletion} given by Aprile et al.~\cite{Aprile2023}.
		
		Let $G$ be a graph and $\alpha \ge 1$.
		Let $H=(W,F)$ be an induced subgraph of $G$ and $w_H \colon W \to \Qpos$ a weight function on $H$.
		Note that $w_H$ is not necessarily related to any weight function on $G$ in any way.
		We will say that $(H,w_H)$ is \emph{strongly $\alpha$-good} in $G$ if $w_H$ is not uniformly $0$ and $w_H(W) \le \alpha \cdot \OPT_{\mathrm{CVD}}(H,w_H)$.
		We will call $(H,w_H)$ \emph{centrally $\alpha$-good} in $G$ with respect to $v_0 \in W$ if \begin{inparaenum}[(i)]
			\item $N_G[v_0] \subseteq W$,
			\item $w_H(v) \ge 1$ for all vertices $v \in N_G[v_0]$, and
			\item $w_H(W) \le \alpha \cdot \OPT_{\mathrm{CVD}}(H,w_H) +1$.
		\end{inparaenum} 
	
		\begin{theorem}[Aprile et al.~\cite{Aprile2023}]
			\label{thm:Aprile}
			There is a polynomial-time algorithm that given a true twin-free graph~$G$ outputs an induced subgraph $H=(W,F)$ of $G$ and a weight function $w_H\colon W \to \Qpos$ on $H$ such that $(H,w_H)$ is strongly $2$-good in $G$ or centrally $2$-good in $G$.
		\end{theorem}
	
		\begin{algorithm}[t]
			\DontPrintSemicolon
			\SetKwInOut{Input}{input}\SetKwInOut{Output}{output}
			\Input{A graph $G=(V,E)$ with a weight function $w \colon V \to \Qpos$}
			\Output{A vertex cover $X$ of $G$ with weight at most $\OPT_{\mathrm{WVC}}(G,w) + \OPT_{\mathrm{CVD}}(G,w)$}
			\If{$G$ is edgeless}
			{
				\Return{$\emptyset$} \label{l:output-base}\; 
			}
			\If{$G$ contains a vertex $v$ with $w(v)=0$}{
				$G' \leftarrow G  - v$\;
				Call this algorithm on $(G',w)$ to obtain a vertex cover $X'$\;
				\eIf{$X'$ is a vertex cover in $G$}{\Return{$X'$}  \label{l:output-w0-notadded}}{\Return{$X' \cup \{v\}$}  \label{l:output-w0-added}}
			}
			\If{$G$ contains true twins $v$ and $v'$ with $w(v) \le w(v')$\label{l:ttwins1} }
			{
				$G' \leftarrow G - v$\;\label{l:ttwins2} 
				Call this algorithm on $(G',w)$ to obtain a vertex cover $X'$\;
				\Return{$X' \cup \{v\}$}\;  \label{l:output-twins}
			}
			Apply the algorithm in \cref{thm:Aprile} to find a weighted induced subgraph $(H=(W,F),w_H)$ of $G$ that is strongly $2$-good or centrally $2$-good in $G$.\label{l:cvd-call-lemma}\;
			$\lambda ^* \leftarrow \max \{ \lambda \in \Qpos \mid \forall v \in W \colon w(v) - \lambda w_H(v) \ge 0\}$\;
			$w'(v) \leftarrow w(v) - \lambda^* w_H(v)$ for all $v\in W$\;
			$w'(v) \leftarrow w(v)$ for all $v \in V\setminus W$\;
			Call this algorithm on $(G,w')$ to obtain a vertex cover $X'$\;
			\Return{$X'$}  \label{l:output-decomposition}
			\caption{}
			\label{alg:cvd}
		\end{algorithm}
	
		Our efficient parameterized approximation algorithm for \textsc{Vertex Cover} parameterized by $\OPTCVD$ is presented in \Cref{alg:cvd}.
		
		\begin{theorem}
			\label{thm:vc-cvd}
			There is a polynomial-time algorithm that, given a weighted graph $(G,w)$, computes a vertex cover of weight at most $\OPTWVC(G,w) + \OPTCVD(G,w)$.
		\end{theorem}
		\begin{proof}
			Let $X$ be the output of \Cref{alg:cvd}.
			It is easy to see that $X$ is a feasible vertex cover.
			We must show that $w(X) \le \OPT_{\mathrm{WVC}}(G,w) + \OPT_{\mathrm{CVD}}(G,w)$.
			
			Before we prove this, we will show that $X$ is an inclusion-wise minimal vertex cover by induction on the number of recursive calls of the algorithm.
			If the algorithm terminates in line \ref{l:output-base}, then this is obvious.
			If it terminates in line~\ref{l:output-w0-notadded} or~\ref{l:output-w0-added}, then this follows from the fact that $X'$ is, by induction hypothesis, inclusion-wise minimal and $v$ is only added to $X'$ if this is $X'$ is not a vertex cover by itself.
			If the algorithm terminates in line~\ref{l:output-twins}, then $X'$ is by induction hypothesis inclusion-wise minimal in $G'$.
			Hence, there is a vertex $u \in N_{G'}[v']$ that is not contained in $X'$.
			Since $N_{G'}[v'] = N_G(v)$, it follows that $u\in N_G(v)$.
			Therefore, $X' \cup \{v\}$ is inclusion-wise minimal in $G$.
			Finally, if the algorithm terminates in line~\ref{l:output-decomposition}, then the returned vertex cover is inclusion-wise minimal by induction.
			
			By induction on the recursion of \Cref{alg:cvd}, we will now show that $w(X) \le \OPT_{\mathrm{WVC}}(G,w) + \OPT_{\mathrm{CVD}}(G,w)$.
			
			In the base case, the algorithm terminates in line \ref{l:output-base}. Clearly, $X$ is an optimal solution in this case, i.e., $w(X) = \OPT_{\mathrm{WVC}}(G,w)$.
			If the algorithm terminates in line~\ref{l:output-w0-notadded} or~\ref{l:output-w0-added}, then $\OPT_{\mathrm{WVC}}(G') \le \OPT_{\mathrm{WVC}}(G)$ and $\OPT_{\mathrm{CVD}}(G') \le \OPT_{\mathrm{CVD}}(G)$.
			Then, by the induction hypothesis $w(X) = w(X') \le \OPT_{\mathrm{WVC}}(G',w) + \OPT_{\mathrm{CVD}}(G',w) =  \OPT_{\mathrm{WVC}}(G,w) + \OPT_{\mathrm{CVD}}(G,w)$.
			If the algorithm terminates in line~\ref{l:output-twins},
			let $X^*$ be a minimum vertex cover in $G$.
			Because $v$ and $v'$ are adjacent, $X^*$ must contain at least one of $v$ or $v'$.
			If $X^*$ contains $v'$ but not $v$, then $(X^*\setminus \{v'\}) \cup \{v\}$ is a vertex cover of smaller weight, contradicting the optimality.
			Hence, $v \in X^*$, implying that $X^*\setminus \{v\}$ is an optimal vertex cover in $G'$.
			Then,
			\begin{align*}
			w(X) &= w(X') + w(v) \le \OPT_{\mathrm{WVC}}(G',w) + \OPT_{\mathrm{CVD}}(G',w) + w(v) \\
			& = w(X^*\setminus \{v\}) + \OPT_{\mathrm{CVD}}(G',w) + w(v) = w(X^*) + \OPT_{\mathrm{CVD}}(G',w) \\
			& = \OPT_{\mathrm{WVC}}(G',w) + \OPT_{\mathrm{CVD}}(G',w) \le \OPT_{\mathrm{WVC}}(G,w) + \OPT_{\mathrm{CVD}}(G',w).
			\end{align*}
			
			Finally, suppose that the algorithm terminates in line~\ref{l:output-decomposition}.
			Clearly, $w' + (w - w') = w$.
			Therefore, in order to apply \cref{lemma:local-ratio}, we must show that $w'(X) \le \OPTWVC(G,w') + \OPT_\mathrm{CVC}(G,w')$
			and $w''(X) \le \OPTWVC(G,w'') + \OPT_\mathrm{CVC}(G,w'')$ where $w'' \coloneqq w-w'$.
			For $w'$, this follows by the induction hypothesis.
			Observe that $w''(v) = \lambda^* w_H(v)$, if $v \in W$, and $w''(v) = 0$, otherwise.
			First, assume that $(H,w_H)$ is strongly $2$-good.
			Then,
			\begin{align*}
				w''(X) & = \lambda^* w_H(X \cap W) \le \lambda^*w_H(W) \le 2 \lambda^* \OPT_{\mathrm{CVD}}(H,w_H)\\
				& \le \lambda^*\OPT_{\mathrm{CVD}}(H,w_H) + \lambda^*\OPT_{\mathrm{WVC}}(H,w_H)\\
				& = \OPT_{\mathrm{CVD}}(H,\lambda^*w_H) + \OPT_{\mathrm{WVC}}(H,\lambda^*w_H)\\
				& = \OPT_{\mathrm{CVD}}(G,w'') + \OPT_{\mathrm{WVC}}(G,w'').
			\end{align*}
			Now, assume that $(H,w_H)$ is centrally $2$-good with respect to $v_0$.
			Because $X$ is inclusion-wise minimal, it cannot contain all of $N[v_0]$.
			Moreover, since $w_H(v) \ge 1$ for all $v\in N[v_0]$, it follows that $w_H(X \cap W) \le w_H(W) - 1$.
			Hence,
			\begin{align*}
				w''(X) &= \lambda^* w_H(X \cap W) \le \lambda^*(w_H(W) - 1) \le 2 \lambda^* \OPT_{\mathrm{CVD}}(H,w_H)\\
				&\le \lambda^*\OPT_{\mathrm{CVD}}(H,w_H) + \lambda^*\OPT_{\mathrm{WVC}}(H,w_H)\\
				& = \OPT_{\mathrm{CVD}}(H,\lambda^*w_H) + \OPT_{\mathrm{WVC}}(H,\lambda^*w_H)\\
				& = \OPT_{\mathrm{CVD}}(G,w'') + \OPT_{\mathrm{WVC}}(G,w'').\qedhere
			\end{align*}
		\end{proof}
	
		Recall that a graph is a \emph{cocluster graph} if its complement is a cluster graph.
		The algorithm in \cref{thm:vc-cvd} can easily be adapted to the parameter $\OPTCCVD$ where \textsc{Cocluster vertex deletion} (CCVD) is the problem of computing a minimum-weight modulator to a cocluster graph.
	
		\begin{corollary}
			\label{cor:vc-ccvd}
			There is a polynomial-time algorithm that, given a weighted graph $(G,w)$, computes a vertex cover of weight at most $\OPTWVC(G,w) + \OPTCCVD(G,w)$.
		\end{corollary}
		\begin{proof}
			We use \Cref{alg:cvd} and the same argument as in the proof of \cref{thm:vc-cvd} with the following modifications.
			We wish to call \cref{thm:Aprile} on the complement graph of $G$.
			Hence, we replace lines \ref{l:ttwins1}--\ref{l:output-twins} by the following.
			The condition in the if-statement is that $G$ contains false twins $v$ and $v'$.
			The graph $G'$ is obtained by replacing $v$ and $v'$ with a single vertex $v''$ with weight $w(v'') \coloneqq w(v) + w(v')$.
			In line~\ref{l:output-twins}, we return $X'$, if $v'' \notin X'$, and $X' \cup \{v,v'\}$, otherwise.
			The correctness of this modification follows from easily from the fact that any inclusion-wise minimal vertex cover in $G$ contains either $N(v) = N(v')$, but not $v$ or $v$', or both $v$ and $v'$.
			In line~\ref{l:cvd-call-lemma}, we must invoke \cref{thm:Aprile} to find a strongly 2-good or centrally 2-good subgraph in $\overline{G}$, the complement graph of $G$, rather than $G$ itself.
			The correctness of this modification follows, just as in the proof of \cref{thm:vc-cvd}, from the fact that $G$ contains a false twin if and only if $\overline{G}$ contains a true twin, the fact that $\OPT_{\mathrm{CCVD}}(G,w) = \OPT_{\mathrm{CVD}}(\overline{G},w)$, and the fact that $\OPT_{\mathrm{CCVD}}(G,w)$ is also a lower bound for  $\OPT_{\mathrm{CVD}}(G,w)$.
		\end{proof}
	
		The results in \cref{thm:vc-cvd} and \cref{cor:vc-ccvd} are tight in the following sense:
		Let $k\in \allowbreak \{\OPTCVD,\allowbreak\OPTCCVD\}$.
		Since an edgeless graph is also both a cluster graph and a cocluster, it follows that $k \le \OPTVC$
		Hence, the algorithms in both results are also $2$-approximations.
		Additionally, computing a vertex cover of weight at most $(1-\eps)\OPTVC + k$ or $\OPTVC + (1-\eps)k$ in polynomial time would imply a $(2-\eps)$-approximation for \textsc{Vertex Cover}, contradicting the Unique Games Conjecture~\cite{Khot2008}.
		
		Both of these results also improve on the the solution quality obtained by applying the structural rounding approach and \cref{thm:struct-rounding}.
		This approach only yields a solution of weight at most $\OPTVC + 2\OPTCVD$ resp. $\OPTVC + 2\OPTCCVD$, since under the Unique Games Conjecture there is no $(2-\eps)$-approximation for CVD or CCVD.
		
		\subsection{Feedback vertex number}
		\label{sec:vc-fvs}
		Now, we consider the parameterization of \textsc{Weighted Vertex Cover} (VC) by the weight of a minimum feedback vertex set.
		A \emph{feedback vertex set} is a modulator to an acyclic graph.
		We will show that \textsc{Vertex Cover} has a polynomial-time approximation that computes a vertex cover of weight at most $\OPTWVC(G) + \OPTFVS(G)$ in any weighted graph $G$.
		
	    Consider the following standard LP relaxation for \textsc{Weighted Vertex Cover}, where $(G=(V,E),w)$ is a weighted graph and $x$ is a vector indexed by $V$:
	    \begin{align*}
	    	&\min \sum_{v\in V} w(v) x_v \\
	    	\text{s.t. }& x_u + x_v \ge 1, \text{ for all } \{u,v\} \in E,\\
	    	& x_v \ge 0, \text{ for all } v \in V.
	    \end{align*}
		Let $\OPT_{\mathrm{LPVC}}(G,w)$ denote the weight of an optimum solution to this LP.
	
		\begin{theorem}[Nemhauser and Trotter~\cite{Nemhauser1974,Nemhauser1975}]
			\label{thm:NT}
			
			\begin{enumerate}[(i)]
				\item There is an optimum solution $x^*$ to the \textsc{Vertex Cover} LP such that $x^*_v \in \{0,\frac{1}{2},1\}$ for all $v\in V$. Such a solution is called \emph{half-integral} and can be computed in polynomial time.
				\item Let $x^*$ be an optimum half-integral solution to the \textsc{Vertex Cover} LP.
				Let $V_i \coloneqq \{v \in V \mid x^*_v = i\}$ for $i \in \{0,\frac{1}{2},1\}$
				Then, there is a minimum-weight vertex cover in $(G,w)$ that contains $V_1$, but does not contain any vertices in $V_0$.
				Setting all variables to $\frac{1}{2}$ yields an optimum LP solution for $(G[V_{\frac{1}{2}}],w)$.
			\end{enumerate}
		\end{theorem}
	
		\begin{theorem}
			\label{thm:vc-fvs}
			There is a polynomial-time algorithm that, given a weighted graph $(G,w)$, computes a vertex cover of weight at most $\OPTWVC(G,w) + \OPTFVS(G,w)$.
		\end{theorem}
		\begin{proof}
			Let $(G=(V,E),w)$ be a given weighted graph.
			Our algorithm uses \cref{thm:NT} to compute the sets $V_0,V_{\frac{1}{2}},V_1 \subseteq V$.
			Let $G' \coloneqq G[V_{\frac{1}{2}}]$.
			By \cref{thm:NT}(ii), it follows that $\OPT_{\mathrm{WVC}}(G,w) = w(V_1) + \OPT_{\mathrm{VC}}(G',w)$.
			We then apply a $2$-approximation algorithm for \textsc{Feedback Vertex Set}~\cite{Bafna1999,Becker1996} to $G'$ to obtain $X\subseteq V_{\frac{1}{2}}$, a feedback vertex set in $G'$ of weight at most $2 \cdot \OPT_{\mathrm{FVS}}(G',w) \le 2 \cdot \OPT_{\mathrm{FVS}}(G,w)$.
			Let $T \coloneqq G' - X$.
			We also compute $Y$, a minimum-weight vertex cover in the forest~$T$.
			We return the solution $V_1 \cup X \cup Y$.
			Clearly, this algorithm can be implemented to run in polynomial time, and always outputs a feasible vertex cover.
			
			To show that the obtained solution always achieves the claimed weight, we first observe that, because $T$ is bipartite, we have $\OPT_{\mathrm{WVC}}(T,w) = \OPT_{\mathrm{LPVC}}(T,w)$.
			By \cref{thm:NT}(ii), the all-$\frac{1}{2}$ vector is an optimum LP solution to for $(G',w)$.
			Since $\OPT_{\mathrm{WVC}}(G',w) \ge \OPT_{\mathrm{LPVC}}(G',w)$, it follows that:
			\begin{align*}
				 \OPT_{\mathrm{WVC}}(T,w) & = \OPT_{\mathrm{LPVC}}(T,w)\\ &\le \OPT_{\mathrm{LPVC}}(G,w) - \frac{w(X)}{2} \le \OPT_{\mathrm{WVC}}(G,w) - \frac{w(X)}{2}.
			\end{align*}
			As a result, we get the following:
			\begin{align*}
				w(V_1 \cup X \cup Y) & = w(V_1) + w(X) + w(Y) \\
				&= w(V_1) + w(X) +  \OPT_{\mathrm{WVC}}(T,w)\\
				&\le w(V_1) + w(X) + \OPT_{\mathrm{WVC}}(G',w) - \frac{w(X)}{2}\\
				&\le w(V_1) + \OPT_{\mathrm{FVS}}(G',w) + \OPT_{\mathrm{WVC}}(G',w)\\
				&= \OPT_{\mathrm{WVC}}(G,w) + \OPT_{\mathrm{FVS}}(G',w)\\
				&\le \OPT_{\mathrm{WVC}}(G,w) + \OPT_{\mathrm{FVS}}(G,w).\qedhere
			\end{align*}
		\end{proof}
	
		Recall that $\OPT_\mathrm{ChVD}$ refers to the weight of a minimum modulator to a chordal graph.
	
		\begin{corollary}
			\label{cor:vc-chordal}
			There is a polynomial-time algorithm that, given a weighted graph $(G,w)$, computes a vertex cover of weight at most $\frac{3}{2}\OPT_\mathrm{WVC}(G,w) + \OPT_\mathrm{ChVD}(G,w)$.
		\end{corollary}
		\begin{proof}
			The algorithm works as follows:
			While the input graph $G$ contains a triangle, it decreases the weight of all three vertices uniformly until at least one of them has weight $0$.
			Vertices with weight $0$ are added to the solution.
			Since any vertex cover contains at least two of the three vertices in the triangle, \cref{lemma:local-ratio} implies that this leads to a solution of size at most $\frac{3}{2} \cdot \OPT_\mathrm{WVC}$.
			Once $G$ is triangle-free, any subgraph of $G$ is chordal if and only if it is a forest.
			Hence, we can apply \cref{thm:vc-fvs} to obtain a vertex cover of size at most $\OPT_\mathrm{WVC} + k$.
		\end{proof}
	
		Just as in the discussion at the end of \cref{sec:vc-cvd}, the result in \cref{thm:vc-cvd} is tight under the Unique Games Conjecture.
		As we discussed at the end of \cref{sec:epa}, \cref{thm:vc-fvs} improves on the result obtained by applying the structural rounding approach.
		The algorithm in \cref{cor:vc-chordal}, on the other hand, may not be.
		We leave open whether there it is possible to compute in polynomial time a vertex cover of weight at most $\OPT_\mathrm{WVC} + \OPT_\mathrm{ChVD}$.
		It is not even clear whether it is possible to achieve $\OPT_\mathrm{VC} + c\OPT_\mathrm{ChVD}$ where $c$ is any constant.
		There is a polynomial-time approximation algorithm that computes a modulator to a chordal graph of size at most $\OPT_\mathrm{ChVD}(G)^{\bigO(1)}$~\cite{Jansen2018}.
		To our knowledge, no constant-factor approximation is known.
		Hence, using the structural rounding approach and \cref{thm:struct-rounding}, we get an algorithm that computes a vertex cover of size at most $\OPTVC +\OPT_\mathrm{ChVD}^{\bigO(1)}$.
			
		\subsection{Split vertex deletion}
		\label{sec:vc-svd}
		We now consider the parameterization of \textsc{Vertex Cover} by the size of a minimum modulator to a split graph.
		Unlike in our previous results, in this section we only consider the unweighted version of the problem.
		We leave open whether our result can be generalized to the weighted case.
		
		Recall that a graph $G=(V,E)$ is a \emph{split graph} if $V$ can be partitioned into two sets $V=C \uplus I$ such that $C$ and $I$ are a clique and an independent set in $G$, respectively.
		The problem of finding a minimum modulator to a split graph is \textsc{Split Vertex Deletion} (SVD).
		
		We will say that a clique $Z\subseteq V$ in a graph $G=(V,E)$ is \emph{$\alpha$-maximal} if it cannot be augmented by a local search within radius $\alpha$, that is, if there are no $Z_1 \subseteq Z$ and $Z_2 \subseteq V \setminus Z$ with $\abs{Z_1} < \abs{Z_2} \le \alpha$ such that $(Z \setminus Z_1) \cup Z_2$ is also a clique.
		
		\begin{lemma}
			Let $G=(V,E)$ be a split graph.
			Let $X\subseteq V$ be a $2$-maximal clique in $G$.
			Then, $G-X$ has a vertex cover of size at most $1$.
			\label{lemma:clique-in-split}
		\end{lemma}
		\begin{proof}
			Let $C$ and $I$ be the clique and the independent set, respectively, in a split representation of $G$.
			If $X=C$, then $G-X = G[I]$ and, thus, $G-X$ has a vertex cover of size $0$.
			Since $X$ is a clique, it may contain at most one vertex from $I$ and, if $X \ne C$, then $C$ must contain exactly one vertex $u \in I$, since otherwise it would not be a maximal clique.
			If there are distinct $v_1,v_2 \in C\setminus X$, then $(X \setminus \{u\}) \cup \{v_1,v_2\}$ is also a clique, contradicting the assumption in the statement.
			Hence, there is at most one vertex $v\in C\setminus X$.
			Then, $\{v\}$ is a vertex cover in $G-X$.
		\end{proof}

		\begin{lemma}
			\label{lemma:vc-2-plus-c}
			For every constant $c$, there is a polynomial-time algorithm that computes a vertex cover of size at most $\max\{\OPTVC,2\OPTVC -c\}$.
		\end{lemma}
		\begin{proof}
			Given a graph $G=(V,E)$, our algorithm iterates over all vertex sets $Y\subseteq V$ of size at most $c$ and applies a standard $2$-approximation to $G-Y$ to obtain a solution $Z_Y$.
			It outputs $Y^* \cup Z_{Y^*}$ where $Y^*$ is the set that minimizes $Y \cup Z_Y$.
			
			Clearly, this algorithm runs in polynomial time for any fixed $c$ and always outputs a feasible vertex cover.
			Let $X^*$ be a minimum vertex cover in $G$.
			If $\abs{X^*} \le c$, then $X^*$ is among the sets chosen as $Y$ in the algorithm.
			The algorithm then determines that $G-Y$ has a vertex cover of size $0$ and outputs $Y$.
			Therefore, we assume that $\abs{X^*} \ge c$.
			Let $Y^*$ be any subset of $X^*$ of size $c$.
			Then, $Y^*$ is among the sets chosen as $Y$ in the algorithm.
			Moreover, $\OPT_{\mathrm{VC}}(G-Y^*) = \OPT_{\mathrm{VC}}(G) - \abs{Y^*}$.
			Hence,
			\begin{align*}
				\abs{Y^* \cup Z_{Y^*}} &= \abs{Y^*} + \abs{Z_{Y^*}} \le \abs{Y^*} + 2\OPT_{\mathrm{VC}}(G-Y^*) = \abs{Y^*} + 2(\OPT_\mathrm{VC}(G) - \abs{Y^*}) \\
				&=2 \OPT_\mathrm{VC}(G) - c.\qedhere
			\end{align*}
		\end{proof}

		\begin{algorithm}[t]
			\DontPrintSemicolon
			\SetKwInOut{Input}{input}\SetKwInOut{Output}{output}
			\Input{A graph $G=(V,E)$}
			\Output{A vertex cover $X$ of $G$ of size at most $\OPT_{\mathrm{VC}}(G) + \OPT_\mathrm{SVD}(G)$.}
			\If{$G$ is edgeless}{\Return $\emptyset$}		
			Find a $2$-maximal clique $Z$ in $G$\;
			\If{$G-Z$ contains a vertex cover $X$ of size at most $1$}{
				\If{there is a vertex $v \in V\setminus Z$ such that $(Z\setminus \{v\}) \cup X$ is a vertex cover in $G$}{\Return $(Z\setminus \{v\}) \cup X$\label{l:vc-cvd-1}}
				\Return $Z\cup X$\label{l:vc-cvd-2}
			}
			Call the algorithm in \cref{lemma:vc-2-plus-c} with $c=2$ on $G-Z$ and store the result in $X'_1$\;\label{l:split-call-lemma}
			Recursively call this algorithm on $G-Z$ and store the result in $X'_2$\;
			$X_i \gets X'_i \cup Z$ for $i\in \{1,2\}$\;
			\Return $X_1$ or $X_2$, whichever is smaller		
			\caption{}
			\label{alg:svd}
		\end{algorithm}
	
		\begin{theorem}
			\label{thm:vc-split}
			There is a polynomial-time algorithm that, given a graph $G$, computes a vertex cover size at most $\OPTVC(G)+ \OPT_\mathrm{SVD}(G)$.
		\end{theorem}
		\begin{proof}
			The algorithm is presented as \Cref{alg:svd}.
			It is clear that this algorithm can be implemented in polynomial time.
			We must show that the solution it outputs has size at most $\OPTVC+ \OPT_\mathrm{SVD}$.
			
			We will prove the claim by induction on the number of recursive calls.
			If $G$ is edgeless, then clearly $\emptyset$ is an optimum solution.
			It is also easy to see that, if the algorithm returns in lines~\ref{l:vc-cvd-1} or~\ref{l:vc-cvd-2}, then it returns an optimum solution.
			We now consider the case where $G$ is not edgeless and $G-Z$ does not contain a vertex cover of size at most $1$.
			Let $M$ be a minimum modulator to a split graph in $G$ and let $X^*$ be a minimum vertex cover in $G$.
			We distinguish two cases, depending on whether or not the chosen $2$-maximal clique $Z$ intersects $M$ or not.
			
			If $Z$ does not intersect $M$, we claim that $\abs{X_1} \le \abs{X^*} + \abs{M}$.
			Because any vertex cover contains all but at most one vertex from any clique, it follows that:
			\begin{align*}
			\abs{X^* \cap Z} \ge \abs{Z} -1.
			\end{align*}
			The clique $Z$ is a $2$-maximal clique in the split graph $G-M$.
			By \cref{lemma:clique-in-split}, $\OPT_\mathrm{VC}(G-(M \cup Z)) \le 1$.
			It follows that:
			\begin{align*}
				\OPT_\mathrm{VC}(G-Z) \le \abs{M} + 1.
			\end{align*}
			Because $X'_1$ is obtained using the algorithm in \cref{lemma:vc-2-plus-c} with $c=2$ and $\OPTVC(G-Z) \ge 2$, it follows that:
			\begin{align*}
				\abs{X'_1} \le 2\OPT_\mathrm{VC}(G-Z) -2
			\end{align*}
			Putting all this together we get that:
			\begin{align*}
				\abs{X^*} + \abs{M} &= \abs{X^* \cap Z} + \abs{X^* \setminus Z} + \abs{M}\\
				&\ge \abs{Z} -1 + \abs{X^* \setminus Z} + \abs{M}\\
				&\ge \abs{Z} -1 + \OPT_\mathrm{VC}(G-Z) + \abs{M} \\
				&\ge \abs{Z} -1 + \OPT_\mathrm{VC}(G-Z) + \OPT_\mathrm{VC}(G-Z) - 1\\
				&=\abs{Z} +2\OPT_\mathrm{VC}(G-Z) -2\\
				&\ge \abs{Z} +\abs{X'_1} = \abs{X_1}
			\end{align*}
			
			Now consider the case where $Z$ does intersect $M$.
			Here, we claim that $\abs{X_2} \le \abs{X^*} + \abs{M}$.
			By induction hypothesis
			\begin{align*}
				\abs{X'_2} &\le \OPT_\mathrm{VC}(G-Z) + \OPT_\mathrm{SVD}(G-Z) \le \OPT_\mathrm{VC}(G) - \abs{Z} + 1 + \OPT_\mathrm{SVD}(G) - 1 \\ &= \OPT_\mathrm{VC}(G) + \OPT_\mathrm{SVD}(G) -\abs{Z}.
			\end{align*}
			Hence,
			\begin{align*}
				\abs{X_2} = \abs{Z} + \abs{X_2'} \le \abs{Z} + \OPT_\mathrm{VC}(G) + \OPT_\mathrm{SVD}(G) - \abs{Z} = \OPT_\mathrm{VC}(G) + \OPT_\mathrm{SVD}(G).
			\end{align*}
		\end{proof}
	
		Just as in the discussion at the end of \cref{sec:vc-cvd}, this result is tight under the Unique Games Conjecture.
		It also beats the solution quality obtained by applying the structural rounding approach and \cref{thm:struct-rounding}.
		The best known approximation algorithms for \textsc{Split Vertex Deletion} achieve a ratio of $2+\eps$ for any $\eps>0$~\cite{Drescher2024,Lokshtanov2020}.
		Hence, with structural rounding, we would get a vertex cover of weight at most $\OPTWVC + (2+\eps)\OPT_\mathrm{SVD}$.
		This approach does, however, have the advantage that it also works in the weighted case.

		\subsection{A lower bound}
		\label{sec:vc-oct}
		In this section, we will show that there is presumably no EPA for \textsc{Vertex Cover} parameterized by the size of an odd cycle transversal with ratio less than $2$ and linear additive error. 
		This is a consequence of the following known approximation lower bound.
		\begin{theorem}[Bansal and Khot~\cite{Bansal2009}]
			\label{thm:Bansal}
			Assuming the Unique Games Conjecture and that $\mathrm{P}\ne \mathrm{NP}$, then for any constant $\eps,\delta>0$ there is no polynomial-time algorithm that, given an $n$-vertex graph containing two disjoint independent each of size $(\frac{1}{2}-\eps)n$, finds an independent set of size at $\delta n$.
		\end{theorem}
		\begin{corollary}
			\label{cor:vc-oct}
			Let  $\alpha < 2$ and $\beta \ge 1$ be given constants.
			Assuming the Unique Games Conjecture and that $\mathrm{P}\ne \mathrm{NP}$, there is no polynomial-time algorithm that, given a graph $G$, computes a vertex cover of size at most $\alpha \OPTVC(G) + \beta \OPT_\mathrm{OCT}(G)$.	\end{corollary}
		\begin{proof}
			Suppose that there is such an approximation and let $G$ be an $n$-vertex graph containing two disjoint independent each of size $(\frac{1}{2}-\eps)n$.
			Then, $\OPT_{\mathrm{OCT}}(G) \le 2\eps n$ and $\OPT_{\mathrm{VC}}(G) \le (\frac{1}{2}+\eps)n$.
			If $X$ is the vertex cover output by the algorithm, then $V \setminus X$ is an independent set of size at least
			\begin{align*}
				n - \abs{X} & \ge n - \alpha \OPT_{\mathrm{VC}}(G) - \beta \OPT_{\mathrm{OCT}}(G) \\
				& \ge n - \alpha(\frac{1}{2} + \eps)n - 2\beta \eps n\\
				& = (1-\alpha(\frac{1}{2} + \eps)-2\beta \eps)n.
			\end{align*}
			Setting $0<\eps<\frac{1-\frac{\alpha}{2}}{\alpha + 2\beta}$ and $0<\delta<1-\alpha(\frac{1}{2} + \eps)-2\beta \eps$ yields a contradiction to \cref{thm:Bansal}.
			Note that both of the upper bounds are positive because $\alpha < 2$.
		\end{proof}

		The best known polynomial-time approximation algorithm for \textsc{Odd Cycle Transversal} is due to Agarwal et al.~\cite{Agarwal2005} and achieves an approximation ratio of $\bigO(\sqrt{\log n})$.
		As observed by Kratsch and Wahlstr\"{o}m~\cite{Kratsch2014}, this can be combined with an FPT algorithm to compute an odd cycle transversal of size at most $\OPT^{3/2}$.
		Combining this with a straightforward application of the structural rounding approach, gives a polynomial time algorithm that computes a vertex cover of size at most $\OPTVC + \OPT_\mathrm{OCT}^{3/2}$.

\section{Connected Vertex Cover}
	\label{sec:cvc}
	In the following, we will give an EPA for \textsc{Connected Vertex Cover} parameterized by the size of a minimum modulator to a split graph.
	This algorithm is similar to the one for \textsc{Vertex Cover} with the same parameter in \cref{sec:vc-svd}.
	Like in that case, we will only consider unweighted graphs.
	\textsc{Weighted Connected Vertex Cover} does not have a constant-factor approximation unless all problems in NP can be solved in time $n^{\bigO(\log\log n)}$~\cite{Fujto2001}.
	Since $\OPT_\mathrm{SVD} \le \OPT_\mathrm{CVC}$, it follows that a polynomial-time algorithm that computes a connected vertex cover of weight at most $\alpha \OPT_\mathrm{CVC} + \beta \OPT_\mathrm{SVD}$ would imply a constant-factor approximation for \textsc{Weighted Connected Vertex Cover}.
	Therefore, such an algorithm is unlikely.
	
	The following algorithm for the unweighted case has superficial similarities with the approximate kernelization for \textsc{Connected Vertex Cover} parameterized by the size of a modulator to a split graph introduced by Krithika et al.~\cite{Krithika2018}.
	Both our algorithm and the approximate kernelization work by contracting cliques to a single vertex, but the details of the algorithms and the underlying arguments for why contracting cliques is effective are fundamentally different.
	
	Given a graph $G=(V,E)$ and a vertex set $Y\subseteq V$, let $G\langle Y\rangle$ be the graph obtained from $G$ by contracting the vertices in $Y$ to a single vertex and appending a leaf to this vertex.
	
	\begin{lemma}
		Let $G=(V,E)$ be a split graph.
		Let $Z\subseteq V$ be a $2$-maximal clique in $G$.
		Then, $G\langle Z \rangle$ has a vertex cover of size at most $2$.
		\label{lemma:cvc-clique-in-split}
	\end{lemma}
	\begin{proof}
		Let $C$ and $I$ be the clique and the independent set, respectively, in a split representation of $G$.
		Let $v$ be the vertex in $G\langle Z \rangle$ that results from contracting $Z$.
		If $X=C$, then $G\langle Z \rangle$ is a star and $\{v\}$ is a vertex cover of size $1$.
		Since $X$ is a clique, it may contain at most one vertex from $I$ and, if $X \ne C$, then $C$ must contain exactly one vertex $u \in I$, since otherwise it would not be a maximal clique.
		If there are distinct $v_1,v_2 \in C\setminus X$, then $(X \setminus \{u\}) \cup \{v_1,v_2\}$ is also a clique, contradicting the assumption in the statement.
		Hence, there is at most one vertex $u\in C\setminus X$.
		Then, $\{u,v\}$ is a vertex cover in $G\langle Z \rangle$.
	\end{proof}

	\begin{lemma}
		\label{lemma:cvc-subset}
		Let $G$ be a connected graph and $X$ a minimum connected vertex cover in $G$. Let $Y \subseteq X$ such that $G[Y]$ is connected.
		Then, $\OPT_\mathrm{CVC}(G)  = \OPT_\mathrm{CVC}(G\langle Y \rangle) + \abs{Y} - 1$.
	\end{lemma}
	\begin{proof}
		Let $v$ be the vertex in $G \langle Y \rangle$ obtained by contracting $Y$.
		
		$\ge$:
		The set $(X \setminus Y) \cup \{v\} $ is a connected vertex cover in $G\langle Y \rangle$ of size at most $\abs{X} - \abs{Y} + 1 = \OPT_\mathrm{CVC}(G) - \abs{Y} + 1$.
		
		$\le$: Let $X'$ be a minimum connected vertex cover in $G\langle Y \rangle$.
		Since $v$ is adjacent to a leaf, $v\in X'$.
		Then, $(X'\setminus \{v\}) \cup Y$ is a connected vertex cover of size at most $\OPT_\mathrm{CVC}(G\langle Y \rangle) + \abs{Y} - 1$ in $G$.
	\end{proof}
	
	\begin{lemma}
		\label{lemma:cvc-2-plus-c}
		For any constant $c$, there is a polynomial-time algorithm that computes a connected vertex cover of size at most $\max\{\OPT_\mathrm{CVC}(G),\OPT_\mathrm{CVC}(G)+\OPT_\mathrm{VC}(G)-c\}$ in any connected graph $G$.
	\end{lemma}
	\begin{proof}
		Given a connected graph $G=(V,E)$, our algorithm iterates over all vertex sets $Y\subseteq V$ of size exactly $c+1$ such that the graph $G[Y]$ is connected. 
		For each such set $Y$, our algorithm applies Savage's $2$-approximation for \textsc{Connected Vertex Cover}~\cite{Savage1982} to the graph $G\langle Y\rangle$.
		Savage's algorithm outputs a connected vertex cover of size at most $\OPT_\mathrm{CVC}(G\langle Y \rangle) + \OPT_\mathrm{VC}(G\langle Y\rangle)$.
		Let $Z_Y$ be the connected vertex cover obtained by inputting $G\langle Y \rangle$.
		Our algorithm outputs $Y^* \cup (Z_{Y^*}\setminus \{v\})$ where $Y^*$ is the set that minimizes $Y \cup (Z_{Y}\setminus \{v\})$ and $v$ is the vertex in $G \langle Y \rangle$ obtained by contracting $Y$.
		
		Clearly, this algorithm runs in polynomial time for any fixed $c$ and always outputs a feasible connected vertex cover.
		Let $X^*$ be a minimum connected vertex cover in $G$ and $Y^*$ be any subset of $X^*$ of size $c$ such that $G[Y]$ is connected.
		Then, $\OPT_{\mathrm{CVC}}(G\langle Y^* \rangle) = \OPT_{\mathrm{CVC}}(G) - \abs{Y^*} + 1$.
		Hence,
		\begin{align*}
		\abs{Y^* \cup (Z_{Y^*} \setminus \{v\})} &= \abs{Y^*} + \abs{Z_{Y^*}} -1 \le \abs{Y^*} + \OPT_{\mathrm{CVC}}(G\langle Y^* \rangle) + \OPT_\mathrm{VC}(G\langle Y^* \rangle) - 1 \\
		&= \abs{Y^*} + (\OPT_\mathrm{CVC}(G) - \abs{Y^*} + 1) + (\OPT_\mathrm{VC}(G) - \abs{Y^*} +1) - 1 \\
		&=\OPT_\mathrm{CVC} +\OPT_\mathrm{VC} - \abs{Y^*} + 1 = \OPT_\mathrm{CVC} +\OPT_\mathrm{VC} - c\qedhere
		\end{align*}
	\end{proof}

	\begin{lemma}
		\label{lemma:cvc-small-sol-after-contraction}
		For any constant $c$, there is a polynomial-time algorithm that computes a minimum vertex cover in any graph $G$ that contains a clique $Z$, which is given to the algorithm as part of the input, such that $\OPT_\mathrm{CVC}(G\langle Z \rangle) \le c$.
	\end{lemma}
	\begin{proof}
		There is a polynomial-time algorithm that computes a minimum vertex cover in any graph $G$ with $\OPT_\mathrm{CVC}(H) \le c+1$, which works by simply brute forcing over all vertex sets of size $c+1$.
		Let $Y$ be a minimum connected vertex cover in $G\langle Z \rangle$.
		Note that $\abs{Y} \le c$.
		Let $v$ be the vertex in $G\langle Z \rangle$ created by contracting $Z$.
		Because $v$ is adjacent to a leaf, $v\in Y$.
		For any $u \in Z$, the graph $G\langle Z\setminus \{u\} \rangle$ contains a connected vertex cover of size at most $c+1$, namely $Y \cup \{u\}$.
		This set induces a connected graph because $Y$ induces a connected graph, $v\in Y$, and $v$ is adjacent to $u$ in $G\langle Z\setminus \{u\} \rangle$.
		
		Our algorithm computes a minimum connected vertex cover $X'_u$ in $G\langle Z\setminus \{u\} \rangle$ for each $u\in Z$ and stores $X_u \coloneqq (Z\setminus \{u\}) \cup (X'_u \setminus \{v\})$ as a potential solution.
		It additionally computes a minimum connected vertex cover $X'$ in $G\langle Z \rangle$ stores $X \coloneqq (X'\setminus \{v\}) \cup Z$ as a potential solution.
		The algorithm outputs the smallest solution encountered that is the smallest set in $\{X\} \cup \{X_u \mid u \in Z\}$.
		
		Clearly, this algorithm can be implemented to run in polynomial time.
		Any vertex cover in $G$ contains either all of $Z$ or all of $Z$ except for a single vertex.
		Let $X^*$ be a minimum connected vertex cover in $G$.
		Suppose that $u \in Z$ is not contained in $X^*$.
		Then, by \cref{lemma:cvc-subset}, $\OPT_\mathrm{CVC}(G) = \OPT_\mathrm{CVC}(G\langle Z\setminus \{u\} \rangle) + \abs{Z} -2 = \abs{X_u}$.
		Hence, $X_u$ is an optimum solution and it is output by the algorithm.
		Now, suppose that $Z\subseteq X^*$.
		Then, by \cref{lemma:cvc-subset}, $\OPT_\mathrm{CVC}(G) = \OPT_\mathrm{CVC}(G\langle Z \rangle) + \abs{Z} -1 = \abs{X}$ and $X$ is output as an optimum solution.
	\end{proof}

	\begin{lemma}
		\label{lemma:cvc-contract-clique}
		Let $G$ be a graph and $Z$ be a clique in $G$. Then,
		\begin{compactenum}[(i)]
			\item $\OPT_\mathrm{CVC}(G\langle Z\rangle) \le \OPT_\mathrm{CVC}(G) - \abs{Z} + 2$
			\item and $\OPT_\mathrm{SVD}(G \langle Z \rangle) \le \OPT_\mathrm{SVD}(G) - 1$, if $\abs{Z} \ge 2$ and $Z$ intersects a minimum modulator to a split graph in $G$.
		\end{compactenum}
	\end{lemma}
	\begin{proof}
		\begin{compactenum}[(i)]
			\item Let $X$ be a minimum connected vertex cover in $G$.
			Because $Z$ is a clique, $\abs{Z\cap X} \ge \abs{Z} -1$.
			Let $X' \coloneqq (X \setminus Z)\cup \{v\}$ where $v$ is the vertex in $G\langle Z\rangle$ that results from contracting $Z$.
			Then, $\abs{X'} = \abs{X} - \abs{Z} + 2 = \OPT_\mathrm{CVC}(G) - \abs{Z} + 2$.
			\item Let $M$ be a minimum modulator to a split graph in $G$.
			If $Z\subseteq M$, then $(M\setminus Z) \cup \{v\}$ is also a modulator to a split graph.
			If $Z \not\subseteq M$, then $M\setminus Z$ is. \qedhere
		\end{compactenum}		
	\end{proof}

	\begin{theorem}
		\label{thm:cvc-split}
		There is a polynomial-time algorithm that, given a graph $G$, computes a connected vertex cover of size at most $\OPT(G)+ \OPT_\mathrm{SVD}(G)$.
	\end{theorem}
		\begin{algorithm}[t]
		\DontPrintSemicolon
		\SetKwInOut{Input}{input}\SetKwInOut{Output}{output}
		\Input{A connected graph $G=(V,E)$}
		\Output{A connected vertex cover of $G$ of size at most $\OPT_{\mathrm{CVC}}(G) + \OPT_\mathrm{SVD}(G)$}
		\If{$G$ consists of a single vertex}{\Return $\emptyset$}
		Find a $2$-maximal clique $Z$ in $G$\;
		\If{$G\langle Z\rangle$ contains a connected vertex cover of size at most $3$}{Use the algorithm in \cref{lemma:cvc-small-sol-after-contraction} to compute an optimum solution $X$\;\label{l:cvc-1}
		\Return{$X$}\;}
		Recursively call this algorithm on $G\langle Z \rangle$ to obtain a solution $X'_1$\;
		Call the algorithm in \cref{lemma:vc-2-plus-c} with $c=4$ on $G\langle Z \rangle$ to obtain a solution $X'_2$\;\label{l:cvc-split-call-lemma}
		$X_i \gets (X'_i \setminus \{v\}) \cup Z$ for $i\in \{1,2\}$ \tcp*{$v$ is the vertex replacing $Z$ in $G\langle Z\rangle $}
		\Return the smaller of $X_1$ and $X_2$
		
		\caption{}
		\label{alg:cvc-svd}
	\end{algorithm}

	\begin{proof}
		The algorithm is presented as \Cref{alg:cvc-svd}.
		It is clear that this algorithm can be implemented in polynomial time.
		We must show that the solution it outputs has size at most $\OPT+ k$.
		
		Let $M$ be a minimum modulator to a split graph in $G$ and let $X^*$ be a minimum connected vertex cover in $G$.
		We will prove the claim by induction on the number of recursive calls.
		For the base case, note that if $G$ consists of a single vertex, then $\emptyset$ is clearly an optimum solution.
		Similarly, if the algorithm terminates in line~\ref{l:cvc-1}, then it outputs an optimum solution.
		We now distinguish two cases depending on whether the clique $Z$ found by the algorithm intersects $M$ or not.
	
		If $Z$ does not intersect $M$, we claim that $\abs{X_1} \le \abs{X^*} + \abs{M}$.
		Because any vertex cover contains all but at most one vertex from any clique, it follows that:
		\begin{align*}
		\abs{X^* \cap Z} \ge \abs{Z} -1.
		\end{align*}
		The clique $Z$ is a $2$-maximal clique in the split graph $G-M$.
		By \cref{lemma:cvc-clique-in-split}, $\OPT_\mathrm{VC}(G-M) \langle Z \rangle) \le 2$.
		It follows that:
		\begin{align*}
		\OPT_\mathrm{VC}(G\langle Z \rangle) \le \abs{M} + 2.
		\end{align*}
		Because $X'_1$ is obtained using the algorithm in \cref{lemma:vc-2-plus-c} with $c=4$, and the fact that $\OPT_\mathrm{CVC}(G\langle Z\rangle) \ge 4$, it follows that:
		\begin{align*}
		\abs{X'_1} \le \OPT_\mathrm{CVC}(G \langle Z \rangle) + \OPT_\mathrm{VC}(G \langle Z \rangle) -4.
		\end{align*}
		Putting all this together we get that:
		\begin{align*}
		\abs{X^*} + \abs{M} &= \abs{X^* \cap Z} + \abs{X^* \setminus Z} + \abs{M}\\
		&\ge \abs{Z} -1 + \abs{X^* \setminus Z} + \abs{M}\\
		&\ge \abs{Z} -1 + \OPT_\mathrm{CVC}(G \langle Z \rangle) -1 + \abs{M} \\
		&\ge \abs{Z} -1 + \OPT_\mathrm{CVC}(G \langle Z \rangle) -1 + \OPT_\mathrm{VC}(G \langle Z \rangle) - 2\\
		&=\abs{Z} +\OPT_\mathrm{CVC}(G \langle Z \rangle) + \OPT_\mathrm{VC}(G \langle Z \rangle) - 4\\
		&\ge \abs{Z} +\abs{X'_1} = \abs{X_1}
		\end{align*}
		
		Now consider the case where $Z$ does intersect $M$.
		Here, we claim that $\abs{X_2} \le \abs{X^*} + \abs{M}$.
		By induction hypothesis and \cref{lemma:cvc-contract-clique},
		\begin{align*}
			\abs{X'_2} &\le \OPT_\mathrm{CVC}(G\langle Z\rangle) + \OPT_\mathrm{SVD}(G\langle Z\rangle) \le \OPT_\mathrm{CVC}(G) - \abs{Z} + 2 + \OPT_\mathrm{SVD}(G) - 1\\
			&= \OPT_\mathrm{CVC}(G) + \OPT_\mathrm{SVD}(G)- \abs{Z} +1.
		\end{align*}
		Hence,
		\begin{align*}
		\abs{X_2} &= \abs{(Z \cup X'_2)\setminus \{v\}} = \abs{Z} + \abs{X'_2} - 1 \le \OPT_\mathrm{CVC}(G) + \mathrm{SVD}(G).
		\qedhere
		\end{align*}
	\end{proof}

	\cref{thm:cvc-split} is tight under the Unique Games Conjecture.
	It also improves on the solution quality that can be obtained by applying the structural rounding approach.
	This approach cannot be applied in the usual manner, since deleting a modulator could lead to a disconnected graph.
	Instead, we have to contract the approximated modulator to a single vertex and then solve the problem exactly on a graph that has a modulator of size $1$ to a split graph.
	\textsc{Connected Vertex Cover} can only be structurally lifted with respect to vertex deletions with a constant $c=2$.
	This is because, in addition to adding the entire modulator to the solution, we must also make the resulting vertex cover connected.
	This is possible while only doubling the size of the modulator (see, e.g., \cite[Lemma~16]{Hols2020}).
	Hence, in this manner we would get a connected vertex cover of size at most $\OPT_\mathrm{CVC} + 4 \OPT_\mathrm{SVD}$.

\section{Coloring}
	\label{sec:col}
	In this section, we will give several efficient parameterized algorithms for the \chromaticnumber problem.
	All of these algorithms compute a coloring that uses at most $\OPTCOL(G-M) + \abs{M}$ colors where $G$ is a given graph and $M$ is a modulator in $G$ to a graph class on which \chromaticnumber can be solved in polynomial time.
	Note that $\OPTCOL(G-M) + \abs{M} \le \OPTCOL(G) + \abs{M}$ and that this is a slightly stronger result than that required in our definition of an EPA.
	
	\subsection{Bounded chromatic number}
	We will start by considering as parameter the size of a modulator to any graph class $\calC$ in which all graphs have a chromatic number bounded by a constant $c$ and on which a $c$-coloring can be found in polynomial time.
	We assume that the $c$-coloring algorithm for $\calC$ may also be applied to graphs that are not in $\calC$, but, when applied to such a graph, it may output a valid or an invalid coloring.
	Whether this coloring is valid or not, can of course be checked in polynomial time.
	Hence, we do not have to require that the graph class $\calC$ be decidable in polynomial time.

	\begin{theorem}
		\label{thm:col-bounded-c}
		Let $\calC$ be a graph class and $c \in \mathbb N$ be a constant such that every graph in $\calC$ is $c$-colorable and there is a polynomial-time algorithm that $c$-colors graphs in $\calC$.
		Then, there is a polynomial-time algorithm that colors any graph with $c + k$ colors where $k$ is the size of a minimum modulator to $\calC$.
	\end{theorem}
	\begin{algorithm}[t]
		\DontPrintSemicolon
		\SetKwInOut{Input}{input}\SetKwInOut{Output}{output}
		\Input{A graph $G=(V,E)$}
		\Output{A $(k+c)$-coloring of $G$ where $k$ is the size of a minimum modulator to $\calC$.}
		Let $V=\{v_1,\ldots,v_n\}$ be an arbitrary ordering of the vertices of $G$\;
		$G_i \gets G[v_1,\ldots,v_i]$ for each $i \in \{1,\ldots,n\}$\;
		$c_0 \gets \bot$ \tcp*{Trivial coloring of the empty graph}
		$\ell \gets 0$ \tcp*{Number of colors used so far}
		\For{$i=1,\ldots,n$}{
			$\mathrm{colored} \gets \mathrm{false}$\;
			\For{ $S \in \binom{\{1,\ldots,\ell\}}{c}$ }{
				$V' \gets \{ v \in \{v_1,\ldots,v_{i-1}\} \mid c_{i-1}(v) \in S\} \cup \{v_i\}$, $G' \gets G[V']$\;
				Attempt to color $G'$ using the $c$-coloring algorithm for $\calC$\;
				\If{this attempt produces a valid $c$-coloring $c'$ of $G'$}{
					$c_i(v) \gets c'(v)$ for all $v\in V'$\;
					$c_i(v) \gets c_{i-1}(v)$ for all $v\in \{v_1,\ldots,v_{i-1}\} \setminus V'$\;
					$\mathrm{colored} \gets \mathrm{true}$\;
					\textbf{Break} inner for-loop\;
				}
			}
			\If{$\mathrm{colored} = \mathrm{false}$}{
				$c_i(v) \gets c_{i-1}(v)$ for all $v\in \{v_1,\ldots,v_{i-1}\}$\;
				$c_i(v_i) \gets \ell +1$\;
				$\ell \gets \ell + 1$\;
			}
		}
		\Return $c_n$
		\caption{}
		\label{alg:k+c-col}
	\end{algorithm}
	\begin{proof}
		The algorithm is presented as \Cref{alg:k+c-col}.
		We start by analyzing its running time.
		The outer for-loop is run for $n$ iterations.
		The inner loop is run for $\binom{\ell}{c} = \bigO(\ell^c)$ iterations where $\ell$ is the number of colors that have been used so far.
		As we will show next, $\ell$ never exceeds $c+k$.
		As a result, we get that the running time is bounded by $\bigO((c+k)^c \cdot n^{\alpha +1})$ where $\bigO(n^\alpha)$ is the polynomial bound on the running time of the $c$-coloring algorithm for $\calC$.
		
		Now, we will prove that this algorithm uses at most $c+k$ colors.
		To this end, let $M$ be a minimum modulator to $\calC$ in a given input graph $G$.
		We will prove, by induction on $i$, that the coloring $c_i$ of $G_i$ uses at most $c+k_i$ colors where $k_i \coloneqq \abs{M \cap \{v_1,\ldots,v_i\}}$ is the number of modulator vertices among the first $i$ vertices according to the arbitrary ordering that the algorithm uses.
		For $i=0$, the claim is trivial as $c_0$ uses $0 \le k + c$ colors.
		
		Suppose that the claim is true for $i-1$.
		Note that, in each iteration of the outer for-loop, the number of colors used increases by at most one.
		Hence, if $c_{i-1}$ uses strictly fewer than $c+k_{i-1}$ colors or if $k_i > k_{i-1}$ (that is, if $v_i \in M$), then the claim is clearly true.
		Therefore, we may assume that $c_{i-1}$ uses exactly $c+k_{i}$ colors.
		It follows that there are is a set $S \subseteq \{1,\ldots,\ell\}$ of at least $c$ colors such that no vertex in $M$ receives a color in $S$ under $c_{i-1}$.
		Let $V' \subseteq \{v_1,\ldots,v_i\}$ contain the vertex $v_i$ and all vertices that receive a color in $S$ under $c_{i-1}$.
		Let $G' \coloneqq G[V']$.
		Note that $M\cap V' = \emptyset$.
		It follows that $G' \subseteq G - M$ and that $G' \in \calC$.
		Therefore, the $c$-coloring algorithm for $\calC$ will be successful in its attempt to color $G'$.
		This means that $c_i$ uses the same number of colors as $c_{i-1}$, i.e. $c+k_i$ colors.
	\end{proof}

	Recall that $\OPT_\mathrm{OCT}$ and $\OPT_\mathrm{PVD}$ refer to the size of a minimum modulator to a bipartite and to a planar graph, respectively.
	Since bipartite graphs can be $2$-colored in polynomial time and planar graphs can be $4$-colored in polynomial time~\cite{Robertson1996}, \cref{thm:col-bounded-c} implies the following:
	
	\begin{corollary}
		\label{cor:col-bounded-c}
		There are polynomial-time algorithms that color any graph with at most
		\begin{compactenum}[(i)]
			\item $2 + \OPT_\mathrm{OCT}$ colors or
			\item $4 + \OPT_\mathrm{PVD}$ colors.
		\end{compactenum}
	\end{corollary}
	
	The results in \cref{cor:col-bounded-c} are tight in the sense that there are graphs that require $2+ \OPT_\mathrm{OCT}$ or $4 + \OPT_\mathrm{PVD}$ colors (consider a graph that contains a modulator that consists of a clique of vertices that are adjacent to all other vertices).
	Like in the discussion in \cref{sec:vc-oct}, we can get a coloring with at most $2+ \OPT_\mathrm{OCT}^{3/2}$
	using structural rounding.
	There is also no known constant-factor approximation for computing a modulator to a planar graph.
	The best known approximation achieves a ratio of $\bigO(\OPT^\eps)$ for any $\eps > 0$~\cite{Jansen2022}.
	Hence, the structural rounding approach yields a coloring with $4 + \OPT_\mathrm{PVD}^{1+\eps}$ colors.
	It is not possible to achieve $\OPTCOL + \OPT_\mathrm{PVD}$ colors, unless $\mathrm{P}=\mathrm{NP}$, because coloring a planar graph optimally is NP-hard~\cite{Garey1976}.
	
	\subsection{Cograph, chordal, and cochordal deletion}
	In the following, we will explore the power of greedy algorithms for EPAs for \chromaticnumber.
	We will show that with two simple greedy strategies, we can get EPAs for this problem parameterized by the size of a modulator to a cograph, a chordal graph, or a cochordal graph.

	The \emph{degeneracy} of a graph $G=(V,E)$ is
	\begin{align*}
		\mathrm{degen}(G) \coloneqq \max_{V'\subseteq V} \min_{v\in V'} \deg_{G[V']}(v).
	\end{align*}
	It is well-known that any graph $G$ can be colored with $\mathrm{degen}(G) + 1$ in polynomial time~\cite{Matula1968}.
	
	\begin{lemma}
		\label{lemma:col-chordal}
		Let $G=(V,E)$ be a graph and $M\subseteq V$ be a modulator to a chordal graph in $G$.
		Then, $\mathrm{degen}(G) \le \OPT_\mathrm{COL}(G-M) + \abs{M} -1$.
	\end{lemma}
	\begin{proof}
		Let $H$ be an induced subgraph of $G$.
		Then, $H-M$ is a chordal graph and, therefore, contains a simplicial vertex $v$.
		Because $v$ is simplicial in $H-M$, $N_{H-M}[v]$ induces a clique in $G$.
		Hence, $N_{H-M}[v] \le \omega(G-M) = \OPT_\mathrm{COL}(G-M)$.
		The degree of $v$ in $H$ is
		\begin{align*}
			\deg_H(v) & = \abs{N_H(v) \cap M} + \abs{N_H(v) \setminus M} = \abs{N_H(v) \cap M} + \abs{N_H[v] \setminus M} -1\\
			&  \le \abs{M} + \OPT_\mathrm{COL}(G-M) -1.\qedhere
		\end{align*}
	\end{proof}

	\begin{theorem}
		\label{thm:col-chordal}
		There is a polynomial-time algorithm that colors any graph $G$ with at most $\OPT_\mathrm{COL}(G-M) + \abs{M}$ colors where $M$ is a minimum modulator to a chordal graph in $G$.
	\end{theorem}
	\begin{proof}
		As noted above, one can color any graph $G$ with $\mathrm{degen}(G) + 1$ colors and, by \cref{lemma:col-chordal}, $\mathrm{degen}(G) + 1 \le \OPT_\mathrm{COL}(G-M) + \abs{M}$.
	\end{proof}

	A graph has the \emph{CK-property} if every maximal clique intersects every maximal independent set in the graph.
	\begin{lemma}[Corneil et al.~\cite{Corneil1981}]
		\label{lemma:cograph}
		A graph is a cograph if and only if it and every one of its induced subgraphs has the CK-property.
	\end{lemma}
	\begin{theorem}
		\label{thm:col-cograph}
		There is a polynomial-time algorithm that colors any graph $G$ with at most $\OPTCOL(G-M) + \abs{M}$ colors where $M$ is a minimum modulator to a cograph in $G$.
	\end{theorem}
	\begin{proof}
		The algorithm picks a maximal independent set $I$, assigns the same color to each vertex in $I$, and deletes $I$ from the input graph.
		This is repeated until the graph is empty.
		
		Let $M$ be a minimum modulator to a cograph in the input graph $G=(V,E)$.
		We must show that the described algorithm uses at most $\OPTCOL(G-M) + \abs{M}$ colors.
		Let $C_1$ be the set of colors that are each assigned to at least one vertex in $M$ and let $C_2$ be all other colors.
		Clearly, $\abs{C_1} \le \abs{M}$, so we will show that $\abs{C_2} \le \OPTCOL(G-M)$.
		
		Let $M'$ be the set of vertices that receive a color in $C_1$ and let $H \coloneqq G- M'$.
		Since $M \subseteq M'$, it follows that $G-M'$ is a cograph and that $\OPTCOL(G-M') \le \OPTCOL(G-M)$,
		Therefore, it is sufficient to show that $\abs{C_2} \le \OPTCOL(G-M')$.
		Let $I_1,\ldots,I_\ell$ be the independent sets picked by the algorithm in that order.
		Let $G_i$ be the graph that remains after $I_1,\ldots,I_i$ have been deleted and let $H_i \coloneqq G_i - M'$.
		
		Suppose that $I_j \in C_2$.
		We claim that $\OPTCOL(H_j) \le \OPTCOL(H_{j-1})-1$, in other words, that deleting $I_j$ decreases the chromatic number of the graph.
		Because $I_j$ is a maximal independent set in the graph $G_{j-1}$ and because $I_j$ does not intersect $M'$, it follows that $I_j$ is a maximal independent set in the cograph $H_{j-1}$.
		By \cref{lemma:cograph}, $I_j$ intersects every maximal clique in $H_{j-1}$.
		This implies $\omega(H_j) \le \omega(H_{j-1})$
		Since $H_{j-1}$ and $H_j$ are cographs, they are perfect, so $\OPTCOL(H_{j-1}) = \omega(H_{j-1})$ ad $\OPTCOL(H_{j}) = \omega(H_j)$.
		Hence, $\OPTCOL(H_j) \le \OPTCOL(H_{j-1}) -1$.
		Inductively it follows that $\abs{C_2} \le \OPTCOL(H)$.
	\end{proof}

	A vertex $v\in V$ is \emph{co-simplicial} if $v$ is simplicial in $\overline{G}$.
	In other words, $v$ is co-simplicial if the set of vertices that are not adjacent to $v$ form an independent set.
	A graph is chordal if and only if it and every one of its subgraphs contains a simplicial vertex.

	We will say that a vertex set $X\subseteq V$ in $G=(V,E)$ is a \emph{separator} if $G-X$ contains more connected components than $G$.
	It is a \emph{coconnected component} (a \emph{coseparator}) if it is a connected component  (a separator) in $\overline{G}$.
	
	\begin{observation}
		\label{obs:col-cochordal}
		Let $G$ be a cochordal graph and let $I$ be a maximal independent set in $G$.
		Then, $I$ is a coseparator in $G$ or there is a minimum coloring of $G$ under which all vertices in $I$ are assigned the same color.
	\end{observation}
	\begin{proof}
		We distinguish two cases depending on whether $I$ contains a cosimplicial vertex or not.
		If $I$ does contain a cosimplicial vertex $v$, let $c$ be any minimum coloring of $G$.
		We obtain $c'$ by assigning the color of $v$ to every vertex in $I$.
		It is easy to see that $c'$ is also a proper coloring and that $c'$ assigns the same color to every vertex in $I$.
		Now, suppose that $I$ does not contain a cosimplicial vertex.
		Then, $I$ does not contain a simplicial vertex in the chordal graph $\overline{G}$.
		It follows that $I$ is a separator in $\overline{G}$.
		Hence, $I$ is a coseparator in $G$.
	\end{proof}

	\begin{theorem}
			\label{thm:col-cochordal}
			There is a polynomial-time algorithm that colors any non-empty graph $G$ with at most $2\OPTCOL(G-M) + \abs{M} -1$ colors where $M$ is a minimum modulator to a cochordal graph in $G$.
	\end{theorem}
	\begin{proof}
		We will show that the same algorithm as in \cref{thm:col-cograph} achieves this.
		The algorithm picks a maximal independent set $I$, assigns the same color to each vertex in $I$, and deletes $I$ from the input graph.
		This is repeated until the graph is empty.
		
		Let $M$ be a minimum modulator to a cochordal graph in the input graph $G=(V,E)$.
		We must show that the described algorithm uses at most $2\OPTCOL(G-M) + \abs{M}$ colors.
		Let $C_1$ be the set of colors that are each assigned to at least one vertex in $M$ and let $C_2$ be all other colors.
		Clearly, $\abs{C_1} \le \abs{M}$, so we will show that $\abs{C_2} \le 2\OPTCOL(G-M)$.
		
		Let $M'$ be the set of vertices that receive a color in $C_1$ and let $H \coloneqq G- M'$.
		We will first consider the edge case where $H$ is the empty graph.
		Since $G$ is non-empty, $G-M$ is non-empty, because, otherwise, $M$ would not be minimum.
		Hence, $\OPTCOL(G-M) \ge 1$.
		However, since every color is assigned to at least one vertex in $M$, it follows that it uses at most $\abs{M} \le 2\OPTCOL(G-M) + \abs{M} -1$ colors.
		
		Now, suppose that $H$ is non-empty.
		Since $M \subseteq M'$, it follows that $G-M'$ is cochordal and that $\OPTCOL(G-M') \le \OPTCOL(G-M)$.
		Therefore, it is sufficient to show that $\abs{C_2} \le 2\OPTCOL(H) -1$.
		By induction on $\abs{C_2}$, we will prove  a stronger claim, namely that $\abs{C_2} \le 2\OPTCOL(H) -r$ where $r$ is the number of coconnected components in $H$ (note that, since $H$ is non-empty, $r\ge 1$).
		
		For the base case, note that, if $\abs{C_2} = 1$, then the claim clearly holds.
		Otherwise, let $I$ be the first independent set contained in $H$ chosen by the algorithm.
		Note that $I$ is a maximal independent set in $H$.
		Hence, by \cref{obs:col-cochordal}, \begin{inparaenum}[(i)]
			\item is a coseparator in $H$ or
			\item there is a minimum coloring of $H$ under which all vertices in $I$ are assigned the same color.
		\end{inparaenum}
		In case~(i), $H-I$ contains at least $r+1$ coconnected components.
		Hence, by induction, the algorithm uses at most $2\OPTCOL(H-I) - r -1$ colors to color $H-I$.
		Hence, taking into account the color assigned to the vertices in $I$, we get that $H$ is colored with at most $2\OPTCOL(H-I) - r \le 2\OPTCOL(H) - r$ colors.
		In case~(ii), $H-I$ contains at least $r-1$ connected components, since an any independent set must be contained within a single coconnected component.
		Moreover, $\OPTCOL(H-I) = \OPTCOL(H) -1$, since $I$ is a color in an optimum coloring of $H$.
		By induction, the algorithm uses at most $2\OPTCOL(H-I) - r +1 = 2\OPTCOL(H) - r -1$ colors to color $H-I$.
		Taking into account the color assigned to the vertices in $I$, this means that the algorithm uses at most $2\OPTCOL(H) - r$ colors to color $H$.
	\end{proof}

	We leave open whether \cref{thm:col-cochordal} can be improved by providing a coloring with $\OPTCOL + \allowbreak \OPT_\mathrm{CChVD}$ colors.
	We will briefly argue that the analysis of the algorithm in \cref{thm:col-cochordal} is tight.
	Consider the complement of the graph in \Cref{fig:col-cochoral-tight}.
	The graph is cochordal, so $\OPT_\mathrm{CChVD} =0$.
	An optimum coloring assigns the same color to $x_i$, $y_i$, and $z_i$ for each $i\in [n]$ for a total of $n$ colors.
	The algorithm may choose the independent sets $\{x_i,y_i\}$ first for $i \in [n-1]$, and then $\{x_1,z_1\}$, $\{y_n,z_n\}$, and $\{z_i\}$ for $i\in \{2,\ldots,n-1\}$ for a total of $2n-1$ colors.
	
	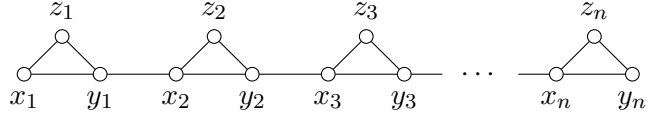
\begin{figure}
		\centering
		\begin{tikzpicture}
			\def\nsc{0.5}
			\def\lheight{-1}
			\tikzstyle{xnode}=[circle,scale=\nsc,draw];
			\node[xnode,label=below:$x_1$] (x1) at (0,0) {};
			\node[xnode,label=below:$y_1$] (y1) at (1,0) {};
			\node[xnode,label=above:$z_1$] (z1) at (0.5,0.5) {};
			\draw (x1) -- (z1) -- (y1) -- (x1);
			
			\node[xnode,label=below:$x_2$] (x2) at (2,0) {};
			\node[xnode,label=below:$y_2$] (y2) at (3,0) {};
			\node[xnode,label=above:$z_2$] (z2) at (2.5,0.5) {};
			\draw (x2) -- (z2) -- (y2) -- (x2);
			\draw (y1) -- (x2);
			
			\node[xnode,label=below:$x_3$] (x3) at (4,0) {};
			\node[xnode,label=below:$y_3$] (y3) at (5,0) {};
			\node[xnode,label=above:$z_3$] (z3) at (4.5,0.5) {};
			\draw (x3) -- (z3) -- (y3) -- (x3);
			\draw (y2) -- (x3);
			
			\draw (y3) -- (5.5,0);
			
			\node () at (6,0) {$\cdots$};
			
			\node[xnode,label=below:$x_n$] (xn) at (7,0) {};
			\node[xnode,label=below:$y_n$] (yn) at (8,0) {};
			\node[xnode,label=above:$z_n$] (zn) at (7.5,0.5) {};
			\draw (xn) -- (zn) -- (yn) -- (xn);
			\draw (6.5,0) -- (xn);
		\end{tikzpicture}
		\caption{An instance showing that the analysis of the algorithm in \cref{thm:col-cochordal} is tight.
		The example is the complement of the pictured graph.}
		\label{fig:col-cochoral-tight}
	\end{figure}

	\subsection[P3+K1-free deletion]{$(P_3+K_1)$-free deletion}
			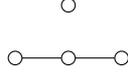
\begin{figure}
		\centering
		\begin{tikzpicture}[scale=.7]
		\def\nsc{0.5}
		\def\lheight{-1}
		\tikzstyle{xnode}=[circle,scale=\nsc,draw];
		\node[xnode] (a) at (1,2) {};
		\node[xnode] (b) at (0,1) {};
		\node[xnode] (c) at (1,1) {};
		\node[xnode] (d) at (2,1) {};
		\draw (b) -- (c) -- (d);
	\end{tikzpicture}
	\caption{The graph $P_3 + K_1$}
	\label{fig:p3k1}
	\end{figure}
	Let $H$ be a graph.
	We say that a graph $G$ is \emph{$H$-free} if $H$ does not occur as an induced subgraph of $G$.
	Kr\'al' et al.~\cite{Kral2001} proved that the Coloring problem is polynomial-time solvable on the class of $H$-free graphs if and only if $H$ is a subgraph of $P_4$ or of $P_3+K_1$ (the graph pictured in \Cref{fig:p3k1}).
	The class of $P_4$-free graphs are precisely the cographs and we showed in \cref{thm:col-cograph} that there is an algorithm that colors any graph with $\OPT_\mathrm{COL}(G-M) + \abs{M}$ colors where $M$ is a minimum modulator to a cograph.
	Here we prove that the same is also to true for a modulator to a $(P_3+K_1)$-free graph.
	
	The \emph{full join} of two graphs $G_1=(V_1,E_1)$ and $G_2=(V_2,E_2)$ is the graph $G_1\times G_2=(V_1 \cup V_2,E')$ where $E'=E_1\cup E_2 \cup \{\{v_1,v_2\}  \mid v_1 \in V_1, v_2\in V_2\}$.
	
	\begin{lemma}[Olariu~\cite{Olariu1988}]
		A graph $G$ is $(P_3+K_1)$-free if and only if $G$ is the full join of graphs $G_1$ and $G_2$ such that $G_1$ is a cocluster graph and $G_2$ is $\overline{K_3}$-free.
	\end{lemma}

	\begin{observation}
		\label{obs:col-k3-free}
		Let $G$ be a $\overline{K_3}$-free graph with $n$ vertices let $X$ be a maximum matching in $\overline{G}$.
		Then, $\OPTCOL(G) = n - \abs{X}$ and there is an optimum coloring of $G$ that assigns the same color to vertices that are matched under $X$ and unique colors to free vertices under $X$.
	\end{observation}
	\begin{proof}
		Let $F$ be the set of free vertices under $X$.
		
		$``\le''$: The coloring that assigns the same coloring to vertices that are matched under $X$ and unique colors to vertices in $F$ is a valid coloring and it uses $\abs{X} + \abs{F} = n - \abs{X}$ colors.
		
		$``\ge''$: Let $C_1,\ldots,C_\ell$ be the color classes in an optimum coloring of $G$.
		Since $G$ is $\overline{K_3}$-free, every $C_i$ contains at most two vertices.
		Suppose that $C_1,\ldots,C_j$ contain two vertices each and $C_{j+1},\ldots,C_\ell$ contain one vertex each.
		If $C_i = \{u,v\}$, $i\le j$, then $\{u,v\}$ is an edge in $\overline{G}$.
		Hence, $C_1,\ldots,C_j$ is a matching in $\overline{G}$.
		Therefore, $j\le \abs{X}$.
		Hence, this coloring uses $\ell = \ell - j + j \le \ell - j + \abs{X} = n -\abs{X}$ colors.
	\end{proof}

	\begin{theorem}
		\label{thm:col-p3k1free}
		There is a polynomial-time algorithm that colors any graph $G$ with at most $\OPT_\mathrm{COL}(G-M) + \abs{M}$ colors where $M$ is a minimum modulator to a $(P_3+K_1)$-free graph in $G$.
	\end{theorem}
	\begin{proof}
		The algorithm proceeds in two phases.
		In the first phase, the algorithm finds a maximal independent set $X$ of size at least $3$, assigns all vertices in $X$ the same color and deletes $X$ from $G$.
		This is repeated as long as the input graph $G$ contains an independent set of size at least $3$.
		Once no such independent set exists, it enters the second phase.
		In the second phase, the algorithm computes a maximum matching $Y$ in $\overline{G}$.
		The free vertices under $Y$ each receive a unique color while pairs of vertices that are matched by $Y$ receive the same color.
		
		Clearly, this algorithm can be implemented in polynomial time.
		We must show that it uses at most $\OPT_\mathrm{COL}(G-M) + \abs{M}$ colors where $M$ is a fixed minimum modulator to a $(P_3+K_1)$-free graph.
		Let $G-M = G_1 \times G_2$ where $G_1$ is a cocluster and $G_2$ is $\overline{K_3}$-free.
		We partition the set of colors used by the algorithm into $C_1$, $C_2$, and $C_3$ as follows. The set $C_1$ contains all colors that contain at least one vertex from $M$.
		The set $C_2$ contains all colors that do not contain a vertex from $M$ and are assigned in the first phase of the algorithm.
		Finally, $C_3$ contains all other colors, that is all colors assigned in the second phase that do not contain a vertex from $M$.
		Of course $\abs{C_1} \le \abs{M}$, so it suffices to show that $\abs{C_2} + \abs{C_3} \le \OPT_\mathrm{COL}(G-M)$.
		Let $H_i$ be the induced subgraph of $G$ containing all vertices that receive a color in $C_i$.

		We start by showing that $\abs{C_2} = \OPTCOL(G-H_1) - \OPTCOL(H_2)$.
		The colors in $C_2$ are each assigned to at least three vertices in $G-M$.
		Therefore, they are only assigned to vertices in~$G_1$.
		Moreover, each such color is a maximal independent set in that cocluster.
		Since $C$ induces a cocluster in $\overline{G-M}$, it induces a cluster graph in $G - M$.
		Cluster graphs are cographs and, therefore, have the CK property.
		It follows that removing a maximal independent set from $C$ reduces its clique number and therefore its chromatic number by $1$.
		Inductively, the claim follows.
		
		Next, we prove that $\abs{C_3} = \OPTCOL(H_3)$.
		Note that $H_3$ does not contain an independent set of size $3$.
		Therefore, by \cref{obs:col-k3-free}, phase two of the algorithm colors $H_3$ optimally.
	\end{proof}

\section{Triangle Packing}
\label{sec:tp}
In this section, we will present two efficient parameterized approximations for the \textsc{Triangle Packing} problem.
These algorithms compute a triangle packing of size at least $\OPTTP(G-M)$ where $G$ is a given graph and $M$ is a modulator to a cluster graph and a cocluster graph, respectively.
We briefly note that this is enough to fulfill the definition of an EPA.

\begin{observation}
	Let $G=(V,E)$ be a graph, $\calT$ be a triangle packing in $G$, and $M\subseteq V$ be a set of vertices.
	Then, $\OPTTP(G-M) \ge \OPTTP(G) - \abs{M}$.
\end{observation}
\begin{proof}
	By induction on $\abs{M}$.
	If $\abs{M} = 0$, then $\OPTTP(G-M) = \OPTTP(G) - \abs{M}$.
	Suppose that the claim holds for $M$ and let $v\in V \setminus M$.
	Let $\calT$ be an optimum triangle packing in $G$.
	Then,
	\begin{align*}
		\OPTTP(G-(M\cup\{v\})) & \ge \OPTTP(G-M) - 1 \ge \OPTTP(G) - \abs{M} - 1\\
		& = \OPTTP(G) - \abs{M\cup \{v\}} \qedhere
	\end{align*}
\end{proof}

For a triangle packing $\calT$, we will say that a vertex is \emph{free} under $\calT$, if it is not contained in any triangle in $\calT$.
We will start by considering the parameterization by the size of a modulator to a cluster graph.

\begin{theorem}
	\label{thm:tp-cvd}
	There is a polynomial-time algorithm that, given a graph $G$, computes a triangle packing of size at least $\OPTTP(G-M)$ where $M$ is a minimum modulator to a cluster graph.
\end{theorem}
\begin{proof}
	We will show that greedily choosing a maximal triangle packing leads to a triangle packing of size at least $\OPTTP(G-M)$.
	
	Suppose that $G-M$ contains $\ell$ clusters $C_1,\ldots,C_\ell$ of sizes $s_1,\ldots,s_\ell$.
	Then, $$ \OPTTP(G-M) = \sum_{i=1}^\ell \lfloor \frac{s_i}{3} \rfloor.$$
	Let $\calT$ be the triangle packing output by the algorithm.
	Under $\calT$, at most two vertices are free in each cluster of $G-M$, since otherwise $\calT$ would not be maximal.
	Hence, for each $i\in [\ell]$, the packing $\calT$ contains at least $\lfloor \frac{s_i}{3} \rfloor$ triangles that intersect $C_i$.
	Since there are no edges between clusters, no triangle can intersect more than one cluster.
	Hence,
	\begin{align*}
		\abs{\calT} = \sum_{i=1}^\ell \abs{\{ T \in \calT \mid T\cap C_i \ne \emptyset\}} \ge \sum_{i=1}^\ell \lfloor \frac{s_i}{3} \rfloor = \OPTTP(G-M). & \qedhere
	\end{align*}
\end{proof}

Next, we will consider the parameterization by the size of a minimum modulator to a cocluster graph.
A triangle packing $\calT$ is $\alpha$-maximal if there are no triangles $T_1,\ldots,T_{\alpha-1}\in \calT$ and $T'_1,\ldots,T'_\alpha \notin \calT$ such that $(\calT \setminus \{T_1,\ldots,T_{\alpha-1}\}) \cup \{T'_1,\ldots,T'_\alpha\}$ is also a triangle packing.

\begin{theorem}
	\label{thm:tp-ccvd}
	There is a polynomial-time algorithm that, given a graph $G$, computes a triangle packing of size at least $\OPTTP(G-M)$ where $M$ is a minimum modulator to a cocluster graph.
\end{theorem}
\begin{proof}
	We will show that any $3$-maximal triangle packing in a cocluster graph is optimal.
	This implies that, in any graph $G$ with a modulator $M$ to a cocluster,
	any $3$-maximal triangle packing has size at least
	$\OPTTP(G-M)$.
	Clearly, a $3$-maximal triangle packing can be computed in polynomial time.
	
	Let $\calT$ be a triangle packing in a cocluster $G=(V,E)$.
	By case distinction, we will prove that either $\calT$ is maximum or it is not $3$-maximal.
	Let $C_1,\ldots,C_\ell$ be the coclusters of $G$, that is, the connected components of $\overline{G}$.
	\begin{compactenum}[(i)]
		\item 	If there are at most two free vertices under $\calT$, then $\calT$ is clearly a maximum packing.
		\item 	If there are three free vertices in distinct coclusters, then $\calT$ is not a maximal packing, since those three vertices form a triangle that can be added to $\calT$.
		\item 	Now consider the case where 
		there are free vertices $v_1$ and $v_2$ that are in the same cocluster $C_i$, while the free vertex $v_3$ is in a different cocluster $C_j$.
		No other cocluster contains a free vertex, as otherwise (ii) applies.
		\begin{compactenum}
			\item 	If there is triangle $T = \{u_1,u_2,_3\} \in \calT$ which does not intersect $C_i$, then at least two of the vertices in $T$, say $u_1$ and $u_2$ are not contained in $C_j$.
			Then, $T'_1 \coloneqq \{v_1,v_3,u_1\}$ and $T'_2 \coloneqq \{v_2,u_2,u_3\}$ are both triangles and $(\calT \setminus \{T\}) \cup \{T'_1,T'_2\}$ is a larger triangle packing proving that $\calT$ is not $3$-maximal.
			\item Now, suppose that $C_i$ intersects every triangle in $\calT$.
			\begin{compactenum}[A.]
				\item If there is a second free vertex $v_4 \ne v_3$ contained in $C_j$ and $C_j$ does not intersect every triangle in $\calT$, then the same argument as in (a) applies with the roles of $C_i$ and $C_j$ reversed.
				\item If there is no additional free vertex in $C_j$, then we will show that $\calT$ is maximum.
				In any triangle packing of $G$, every triangle must contain at least two vertices that are not in $C_i$.
				Hence, $\OPTTP(G) \le \lfloor \frac{\abs{V \setminus C_i}}{2} \rfloor$.
				The packing $\calT$ covers every vertex outside of $C_i$ except $v_3$ and every triangle in $\calT$ covers exactly two such vertices.
				Therefore $\abs{\calP} = \frac{\abs{V\setminus C_i}-1}{2}$.
				Since $\abs{\calP}$ is an integer, it follows that $\abs{\calP} = \frac{\abs{V\setminus C_i}-1}{2} = \lfloor \frac{\abs{V \setminus C_i}}{2} \rfloor \ge \OPTTP(G)$.
				\item Now, we may assume that every triangle intersects both $C_i$ and $C_j$.
				Since any triangle packing uses a vertex outside of both $C_i$ and $C_j$, it follows that $\OPTTP(G) \le \abs{V\setminus (C_i \cup C_j)}$.
				Since every vertex outside of $C_i$ and $C_j$ is covered by $\calT$ and every triangle in $\calT$ covers exactly one such vertex, it follows that $\abs{\calT} = \abs{V\setminus (C_i \cup C_j)}$.
			\end{compactenum}
		\end{compactenum}
		\item Now consider the case where there are free $v_1,v_2,v_3$, all contained in the same cocluster $C_i$ and no vertices outside of $C_i$ are free.
		We distinguish two sub-cases depending on whether every triangle in $\calT$ intersects $C_i$.
		\begin{compactenum}
			\item Suppose that there is a triangle $T=\{u_1,u_2,u_3\} \in \calT$ that does not intersect $C_i$.
			Let $T' \coloneqq \{v_1,u_1,u_2\}$ and
			let $\calT'\coloneqq (T \setminus \{T\}) \cup \{T'\}$.
			We can apply the argument in~(iii), replacing $\calT$ with $\calT'$.
			The argument in (iii) implies that either $\calT'$ is a maximum packing in $G$, which means that $\calT$ is also maximum since $\abs{\calT} = \abs{\calT'}$, or that $\calT$ can augmented by removing one triangle and adding two, in which case $\calT$ can be augmented by removing two triangles and adding three.
			\item If every triangle in $\calT$ intersects $C_i$, then we will show that $\calT$ is maximum.
			(This argument is essentially the same as in (iii)(b).B).
			Any triangle contains at least two vertices outside of $C_i$ and, therefore, $\OPTTP(G) \le \frac{\abs{V\setminus C_i}}{2}$.
			The packing $\calT$ covers all vertices outside of $C_i$ and, therefore, $\abs{\calT} = \frac{\abs{V\setminus C_i}}{2}$.
			\qedhere
		\end{compactenum}
	\end{compactenum}
	
\end{proof}

\section{Conclusion}\label{section:conclusion}

We have initiated a dedicated study of efficient parameterized approximation algorithms with small additive error depending on some structural parameter. This is motivated by the goal of leveraging structure for improved polynomial-time approximation results, similarly to how fixed-parameter tractable algorithms improve upon exact exponential-time algorithms within certain regimes of input structure. As a byproduct of the structural rounding framework of Demaine et al.~\cite{Demaine2019}, such error bounds can be obtained for a variety of hard graph problems relative to modulator-based parameters. By focusing directly on the question of additive errors and, as far as possible, seeking to avoid an explicit (approximate) computation of the relevant modulators, we obtain better error bounds than by straightforward application of the framework. In many cases, further improvement of our bounds would contradict known lower bounds for approximation, or the bounds are tight for certain inputs. Arguably, this kind of efficient parameterized approximation algorithm is a sound alternative to computing exact solutions in time exponential in the parameter via FPT-algorithms.

For future work, more examples of efficient parameterized approximation algorithms with error bound better than via the structural rounding framework are of great interest. Ideally, the (approximate) computation of the parameter/modulator can be avoided. Let us name some specific problems related to our results: \vertexcover relative to the size of a modulator to chordal graphs, graphs of treewidth at most two, or claw-free graphs. \weightedvertexcover and \trianglepacking with modulator to split graphs. \chromaticnumber with modulator to weakly chordal, comparability, or perfect graphs. For problems that resist polynomial/linear error bounds one may also study parameterization by the size of an \emph{edge} modulator (deletion/addition/editing) to the class in question. As an example, is there a linear additive error for \vertexcover with edge modulator to bipartite graphs?

Certainly lower bounds are of great interest as well. Note that these are challenging because one needs to construct inputs with a small but \emph{unknown modulator} to some target class, while being hard to find a good approximation for the target problem. As an example, with a given optimal modulator to bipartite graphs it would clearly be no challenge to beat the lower bound of $(2-\varepsilon)\OPT + c \cdot \OPT_{\mathrm{OCT}}$ for \vertexcover.

\bibliography{str-appr}

\end{document}